\documentclass[reprint]{revtex4-2}
\pdfoutput=1

\usepackage{uniinput}
\usepackage{graphicx} 

\usepackage{comment}
\usepackage{bbold}
\usepackage{amsmath}
\usepackage{amsthm}
\usepackage{mathtools}
\usepackage{amssymb}
\usepackage{multirow}
\newtheorem{definition}{Definition}
\usepackage{braket}
\usepackage[compat=0.6]{yquant}
\usetikzlibrary{fit,quotes}
\usetikzlibrary{shapes.geometric,arrows,positioning,calc}
\yquantset{plusctrl/.style={/yquant/every control/.style={/yquant/operators/every not}, every positive control/.style={}}}

\usepackage{xcolor}
\definecolor{riverlane_green}{RGB}{0, 111, 98}
\definecolor{riverlane_light_green}{RGB}{0, 150, 143}
\definecolor{riverlane_orange}{RGB}{255, 117, 0}
\definecolor{riverlane_red}{RGB}{220, 68, 5}
\definecolor{riverlane_pink}{RGB}{207, 111, 127}
\usepackage{hyperref}
\hypersetup{
  colorlinks   = true, 
  urlcolor     = riverlane_green, 
  linkcolor    = riverlane_orange, 
  citecolor   = riverlane_green  
}

\newcommand{\re}{\mathrm{Re}\,}
\newcommand{\im}{\mathrm{Im}\,}
\newcommand{\ee}{\mathrm{e}}
\newcommand{\ii}{\mathrm{i}}

\newcommand{\T}{\mathbb{T}}
\newcommand{\C}{\mathbb{C}}
\newcommand{\R}{\mathbb{R}}
\newcommand{\D}{\mathbb{D}}
\newcommand{\Z}{\mathbb{Z}}

\def\Xint#1{\mathchoice
{\XXint\displaystyle\textstyle{#1}}%
{\XXint\textstyle\scriptstyle{#1}}%
{\XXint\scriptstyle\scriptscriptstyle{#1}}%
{\XXint\scriptscriptstyle\scriptscriptstyle{#1}}%
\!\int}
\def\XXint#1#2#3{{\setbox0=\hbox{$#1{#2#3}{\int}$}
\vcenter{\hbox{$#2#3$}}\kern-.5\wd0}}

\def\pvint{\Xint-}

\newtheorem{theorem}{Theorem}
\newtheorem{corollary}{Corollary}
\newtheorem{lemma}{Lemma}
\newtheorem{remark}{Remark}

\begin{document}

\title{Generalized Quantum Singular Value Transformation}
\author{Christoph Sünderhauf}
\email{christoph.sunderhauf@riverlane.com}
\affiliation{Riverlane, St.~Andrews House, 59 St.~Andrews Street, Cambridge CB2 3BZ, United Kingdom}
\date{October 2023}

\begin{abstract}
The quantum singular value transformation has revolutionised quantum algorithms. By applying a polynomial to an arbitrary matrix, it provides a unifying picture of quantum algorithms. However, polynomials are restricted to definite parity and real coefficients, and finding the circuit (the phase factors) has proven difficult in practice. Recent work has removed these restrictions and enabled faster computation of phase factors, yet only for unitary matrices.
Here we propose two generalisations. The generalised quantum singular value transformation allows complex polynomials for arbitrary matrices. For Hermitian matrices, we propose the generalised quantum eigenvalue transformation that even allows polynomials of indefinite parity. While we find that the polynomial might have to be downscaled compared to the quantum singular value transformation, the higher expressivity of polynomials and faster computation of phase factors can sometimes result in advantages.
The results are achieved with various block encoding (or projected unitary encoding) techniques, including qubitisation, Hermitianisation, and multiplication. We show how to multiply block-encoded matrices with only one extra qubit, and introduce measure-early multiplication to further avoid the extra qubit and decrease average circuit length.

\end{abstract}

\maketitle

The quantum singular value transformation (QSVT) \cite{gilyenQuantumSingularValue2019} has ushered in a new era of quantum algorithms. It allows application of a polynomial $p(x)$ of degree $d$ to \emph{all} singular values of a matrix, with a quantum circuit having $d$ queries to the matrix. By approximating various target functions by polynomials, it led to a veritable grand unification of quantum algorithms \cite{martynGrandUnificationQuantum2021}, expressing algorithms such as amplitude amplification (Grover) \cite{nielsenQuantumComputationQuantum2010}, Hamiltonian simulation \cite{nielsenQuantumComputationQuantum2010}, phase estimation \cite{nielsenQuantumComputationQuantum2010}, and linear system solvers \cite{harrowQuantumAlgorithmSolving2009} in the QSVT formalism \cite{gilyenQuantumSingularValue2019,martynGrandUnificationQuantum2021}. In the special case of Hermitian matrices, the quantum eigenvalue transformation (QET) applies a polynomial to the eigenvalues.

\begin{table}
\begin{ruledtabular}
\footnotesize
\newcommand{\myspace}{\hspace{1em}}
\begin{tabular}{lclcccl}
 &&\myspace Matrices & \multicolumn{2}{c}{\underline{Polynomials allowed}} & phase factors \\
&&\myspace allowed  & complex & indefinite & complexity \\
&&   & coefficients & parity & for degree $d$ \\
\hline
QSP \cite{lowMethodologyResonantEquiangular2016,gilyenQuantumSingularValue2019} &--- &\multicolumn{4}{l}{ Quantum signal processing} \\
 &&\myspace scalars & no & no & $\tilde O(d^2)$ \\
QET \cite{gilyenQuantumSingularValue2019} &--- & \multicolumn{4}{l}{ Quantum eigenvalue transformation} \\
 &&\myspace Hermitian & no & no & $\tilde O(d^2)$ \\
{QSVT \cite{gilyenQuantumSingularValue2019}} &--- & \multicolumn{4}{l}{Quantum singular value transformation}\\
&&\myspace arbitrary & no & no & $\tilde O(d^2)$ \\
{GQSP \cite{motlaghGeneralizedQuantumSignal2023a}} &--- & \multicolumn{4}{l}{Generalised quantum signal processing} \\
&&\myspace unitary & yes & yes & $\tilde O(d)$ \\
GQET& --- & \multicolumn{4}{l}{Generalised quantum eigenvalue transformation} \\
\ [here] & &\myspace Hermitian & yes & yes & $\tilde O(d)$ \\
GQSVT &--- & \multicolumn{4}{l}{Generalised quantum singular value transformation}\\
\ [here]& &\myspace arbitrary & yes & no & $\tilde O(d)$ \\
\end{tabular}
\end{ruledtabular}
\caption{Overview of different polynomial matrix transformations and acronyms.} \label{tablemain:methods}
\end{table}

Yet, QSVT and QET only allow polynomials with real coefficients and definite parity (see Table~\ref{tablemain:methods} for an overview of acronyms and limitations). General polynomials must be implemented with linear combination of unitaries, incuring overhead and resulting in undesirable down-scaling of the polynomial \cite{gilyenQuantumSingularValue2019}. Moreover, the practical difficulty of computing phase factors required for the quantum circuit (the best algorithm approximately requires quadratic time $\tilde O(d^2)$ \cite{dongRobustIterativeMethod2023}) poses a threat to QSVT's applicability and has led to alternatives being considered \cite{steudtnerFaulttolerantQuantumComputation2023}. These limitations are inherited from quantum signal processing (QSP) \cite{lowMethodologyResonantEquiangular2016,gilyenQuantumSingularValue2019}. Recently, generalised quantum signal processing (GQSP) \cite{motlaghGeneralizedQuantumSignal2023a} was introduced to alleviate these limitations, allowing complex coefficients, indefinite parity, and faster computation of phase factors in almost linear time $\tilde O(d)$ \cite{motlaghGeneralizedQuantumSignal2023a}. However, it is only applicable to unitary matrices. 

Here, we lift the restriction of unitary matrices of GQSP and develop the generalised quantum eigenvalue transformation (GQET) for Hermitian matrices and two variants of the generalised quantum singular value transformation (GQSVT) for arbitrary matrices, see the overview in Fig.~\ref{figmain:overview}. They inherit the advantages that GQSP enjoys over QSP, see Table~\ref{tablemain:methods}. However, we find that for GQET and GQSVT, polynomials might have to be scaled down compared to QET and QSVT. Numerical checks for quantum matrix inversion indicate that this is not a severe issue.

We will present the techniques for block encodings that underlie our derivations for GQET and two variants of GQSVT, see Fig.~\ref{figmain:overview}. These are qubitisation \cite{poulinQuantumAlgorithmSpectral2018,berryImprovedTechniquesPreparing2018}, Hermitianisation \cite{chakrabortyPowerBlockencodedMatrix2019}, and multiplication, for which we propose an improved circuit with only one extra qubit and a \emph{measure-early} approach to reduce average circuit length.

\begin{figure}
\centering
\begin{tikzpicture}[xscale=1, yscale=1, every node/.style={font={\footnotesize}}]

\tikzstyle{arrow} = [thick,->,>=stealth,shorten >=2pt]
\tikzstyle{qsp} = [rectangle, rounded corners, align=center, draw=black]
\tikzstyle{be} = [rectangle, align=center,fill=riverlane_light_green!20]

\node (gqsp) at (-2.2, 0) [qsp] {\textbf{GQSP}\\ Transf.~by $P(z) = \sum_n a_n z^n$ \\ for unitary matrices};
\node (qubitisation) at (1.7, 0) [be] {\textbf{Qubitisation} \\ Encoding $U$, reflection $\mathcal R$,\\ qubitised unitary $\mathcal R U$};
\node (gqet) at (-0.5,-1.6) [qsp] {\textbf{GQET}\\ Transformation by $p(x) = \sum_n a_n T_n(x)$\\ for Hermitian matrices \\ with GQSP of qubitised unitary $\mathcal{R} U$};
\node (hermitianise) at (-2.8, -3.4) [be] {\textbf{Hermitianisation}\\
Hermitian encoding of\\ $\bar A/α = \begin{pmatrix} & A/α \\ A^\dag/α & \end{pmatrix}$};
\node (multiply) at (2.1, -3.4) [be] {\textbf{Multiplication}\\ Encoding of $A_1A_2/(α_1α_2)$\\ with 1 extra qubit};
\node (gqsvtherm) at (-2.2, -5) [qsp] {\textbf{GQSVT via Hermitianisation}\\ with GQET of Hermitianised $\bar A$};
\node (gqsvtmult) at (2.2, -5) [qsp] {\textbf{GQSVT via Multipl.}\\ with GQET of product $A^\dag A$};
\draw [arrow] (gqsp) -- (gqet);
\draw [arrow] (qubitisation) -- (gqet);
\draw [arrow] (hermitianise) -- (gqsvtherm);
\draw [arrow] ($(gqet.south)-(0,0)$) -- node[anchor=west,pos=0.8] {(*)} ($(gqsvtherm.north)+(0.7,0)$);
\draw [arrow] ($(gqet.south)-(0,0)$) -- ($(gqsvtmult.north)-(1.5,0)$); 
\draw [arrow] (multiply) -- (gqsvtmult);
\end{tikzpicture}
\caption{Overview of results. We develop the \protect\tikz [anchor=base, baseline]\protect\node [rectangle, rounded corners, align=center, draw=black] {\footnotesize polynomial matrix transformations}; GQET (generalised quantum eigenvalue transformation) for Hermitian matrices and GQSVT (generalised quantum singular value transformation) for general matrices from GQSP (generalised quantum signal processing \cite{motlaghGeneralizedQuantumSignal2023a}) using various
\protect\tikz [anchor=base, baseline]\protect\node [rectangle, align=center,fill=riverlane_light_green!20] {\footnotesize block encoding techniques};, expressed in projected unitary encodings in appendix~\ref{sec:encoding techniques}. Each of the two GQSVT methods can be advantageous in some situations. The method (*) to construct GQSVT from GQET using Hermitianised matrix encodings also offers a new derivation of the usual QSVT circuit from QET, see appendix~\ref{sec:qsvt from qet}. 
} \label{figmain:overview}
\end{figure}

\emph{Block Encodings---}
As quantum computation is unitary, a general matrix $A\in\mathbb{C}^{N_L\times N_R}$ cannot be directly implemented in a quantum circuit. Instead, it can be embedded in the top-left block of a larger unitary \emph{block encoding} $U\in U(M)$, from which the subnormalised $A/α$ (subnormalisation $α$) can be recovered by initialising and postselecting ancilla qubits as $\ket{0}$:
\begin{equation}
    U = \begin{pmatrix} A/α & B \\ C & D\end{pmatrix},\ 
    \vcenter{\hbox{%
    \begin{tikzpicture}
    \begin{yquant}
    qubit {$\ket{0}$} flag;
    qubit {$\ket{j}$} block;
    ["north:$M/N_R$" {font=\protect\footnotesize, inner sep=0pt}]
    slash flag;
    ["north:$N_R$" {font=\protect\footnotesize, inner sep=0pt}]
    slash block;
    box {$U$} (flag, block);
    ["north:$M/N_L$" {font=\protect\footnotesize, inner sep=0pt}]
    slash flag;
    ["north:$N_L$" {font=\protect\footnotesize, inner sep=0pt}]
    slash block;
    output {$\ket{i}$} block;
    output {$\ket{0}$} flag;
    \end{yquant}
    \end{tikzpicture}}}
    = A_{ij}/α
    \label{eqmain:block encoding}
\end{equation}
 For ease of presentation, we take all dimensions $N_L,N_R,M$ powers of 2 throughout. More generally, we may express all our results in terms of the more general \emph{projected unitary encodings}, see the appendices for a detailed exposition and derivation.
 
\emph{Qubitisation---} Hermitian matrices $A\in\mathbb{C}^{N\times N}$ admit a block encoding such that $U$ \eqref{eqmain:block encoding} is also Hermitian. Often, Hermeticity can be achieved by direct construction of the quantum circuit for $U$ without overhead \cite{sunderhaufBlockencodingStructuredMatrices2023}. The qubitized operator \cite{poulinQuantumAlgorithmSpectral2018,berryImprovedTechniquesPreparing2018}
\begin{equation}
    \mathcal R U\ \text{with}\ \mathcal R = \begin{pmatrix}\mathbb{1}_{N\times N} &0 \\ 0 & -\mathbb{1}_{\frac{M}{N}\times \frac{M}{N}}\end{pmatrix}
\end{equation}
is defined as a product of the reflections $U$ and $\mathcal R$ around the coding subspace. It turns out both can be decomposed into a direct sum of reflections for each eigenvalue $\lambda_i$ of $A$ ($A\vec\lambda_i=\lambda_i\vec\lambda_i$), then Jordan's lemma shows that the qubitized operator is a direct sum of rotations, each with eigenvalues $\Lambda_i^\pm = e^{\pm i \arccos(\lambda_i/α)}$. Therefore, qubitisation is widely used in quantum chemistry together with phase estimation to find the (arcosines of the) eigenvalues of a block-encoded Hamiltonian matrix \cite{lowHamiltonianSimulationQubitization2019,vonburgQuantumComputingEnhanced2021,leeEvenMoreEfficient2021,ivanovQuantumComputationPeriodic2023a}. Here we use the same primitive but for a different purpose. An explicit expression for the corresponding eigenvectors $\vec\Lambda_i^\pm$ of $\mathcal{R}U$ using the lower left block of \eqref{eqmain:block encoding} is
\begin{equation}
    \vec\Lambda_i^\pm = \frac{1}{\sqrt{2}}\begin{pmatrix} \vec\lambda_i \\ \frac{\pm i C}{\sqrt{1-(λ_i/α)^2}}\vec\lambda_i\end{pmatrix}.
\end{equation}
In the cases $λ_i/α=\pm 1$, the second block is 0. While $\mathcal{R}U$ has further eigenvectors, these span the coding subspace: 
\begin{equation}
\begin{pmatrix}\vec\lambda_i \\ 0\end{pmatrix} = \frac{1}{\sqrt{2}}\left(\vec\Lambda_i^+ + \vec\Lambda_i^-\right).
\label{eqmain:coding subspace}
\end{equation}
An anti-controlled qubitised operator can be implemented as a quantum circuit:
\begin{equation}
\begin{tikzpicture}
\begin{yquant}
qubit {} cont;
qubit {} flag;
qubit {} block;
slash flag;
["north:$N$" {font=\protect\footnotesize, inner sep=0pt}]
slash block;
box {$\mathcal R U$} (flag,block) ~ cont;
text {$=$} (-);
slash flag;
slash block;
box {$U$} (flag, block) ~ cont;
box {$-Z$} cont ~ flag;
box {$-Z$} cont;
\end{yquant}
\end{tikzpicture}
\label{eqmain:qubitisation circuit}
\end{equation}

\emph{Generalised Quantum Eigenvalue Transformation (GQET)---}
For eigenvalue transformation of $A=A^\dag$, we apply a GQSP sequence to the qubitised unitary operator $\mathcal R U$. A GQSP sequence \cite{motlaghGeneralizedQuantumSignal2023a} consists of rotations
\begin{equation}
    R(θ_i, \phi_i, λ) =\begin{pmatrix}e^{i(λ+\phi_i)}\cosθ_i & e^{i\phi_i}\sinθ_i \\ e^{iλ}\sinθ_i & -\cosθ_i\end{pmatrix}
\end{equation}
by phase factors $(\{\theta_i\}, \{\phi_i\}, λ), i=0\ldots d$, interleaved with the anti-controlled unitary. The $\text{GQET}[U]= \text{GQSP}[\mathcal R U]$ circuit is shown in Fig.~\ref{figmain:gqet circuit}, where the $-Z$ gates from \eqref{eqmain:qubitisation circuit} have been absorbed into $\phi_i'=\phi_i+\pi$ for $i\neq d$.

\begin{figure*}
\begin{tikzpicture}
\begin{yquant}
qubit {$\ket{0}$} flag1;
qubit {$\ket{0}$} flags;
qubit {$\ket{j}$} block;
["north:$N$" {font=\protect\footnotesize, inner sep=0pt}]
slash block;
["north:$M/N$" {font=\protect\footnotesize, inner sep=0pt}]
slash flags;
%
box {$R(\theta_0, \phi'_0, λ)$} flag1;
box {$U$} (block, flags) ~ flag1;
[operator style={/yquant/every negative control}]
phase {} flag1 ~ flags;
box {$R(\theta_1, \phi_1', 0)$} flag1;
box {$U$} (block, flags) ~ flag1;
[operator style={/yquant/every negative control}]
phase {} flag1 ~ flags;
text {$\cdots$} -;
box {$R(\theta_{d-1}, \phi_{d-1}', 0)$} flag1;
box {$U$} (block, flags) ~ flag1;
[operator style={/yquant/every negative control}]
phase {} flag1 ~ flags;
box {$R(\theta_{d}, \phi_{d}, 0)$} flag1;
output {$\ket{0}$} flag1, flags;
output {$\ket{i}$} block;
\end{yquant}
\end{tikzpicture}
\caption{Circuit for generalised quantum eigenvalue transformation (GQET) of a Hermitian matrix $A/α$ in a Hermitian block encoding $U$ \eqref{eqmain:block encoding} by a polynomial $p(x)=\sum_{n=0}^d a_n T_n(x)$. The phase factors $(\{\theta_i\}, \{\phi_i=\phi_i'-\pi\}, \lambda)$ are those determined by GQSP for $P(z) = \sum_{n=0}^d a_nz^n$.}\label{figmain:gqet circuit}
\end{figure*}

As shown in \cite{motlaghGeneralizedQuantumSignal2023a}, phase factors can be found such that the GQSP circuit effects a transformation of the unitary's eigenvalues by an arbitrary complex polynomial
\begin{equation}
    P(z) = \sum_{i=0}^d a_n z^n\ \text{with}\ \max_{|z|=1}|P(z)|\le1
    \label{eqmain:gqsp polynomial}
\end{equation}
on the unit circle. The GQET circuit then block-encodes
\begin{equation}
    p(A/α),\ p(x)=\sum_{i=0}^d a_nT_n(x),
    \label{eqmain:gqet polynomial}
\end{equation}
with the Chebyshev polynomials $T_n(x)$. 
This follows from the fact that GQSP[$\mathcal R U$] applies $P(z)$ to the eigenvalues $\Lambda^{\pm}_i$ of $\mathcal RU$: Initialising and postselecting the top qubit from Fig.~\ref{figmain:gqet circuit} as $\ket{0}$, we have
\begin{align}
(\vec\Lambda^{\pm}_j)^\dag \text{GQSP}[\mathcal R U] \vec\Lambda^{\pm}_i & = P(e^{\pm i \arccos(\lambda_i/\alpha)} )δ_{ij}, \label{eqmain:gqsp} \\
(\vec\Lambda^{\pm}_j)^\dag \text{GQSP}[\mathcal R U] \vec\Lambda^{\mp}_i & = 0\ \text{(orthonormality of}\ \vec\Lambda^{\pm}_i).
\end{align}
Therefore, using \eqref{eqmain:coding subspace}, the matrix elements of the $\text{GQET}[U]=\text{GQSP}[\mathcal R U]$ circuit are
\begin{align}
&\begin{pmatrix}\vec\lambda_j \\ 0 \end{pmatrix}^\dag\text{GQET}[U]\begin{pmatrix}\vec\lambda_i \\ 0 \end{pmatrix} \\
&=\frac{1}{2}\left(P(e^{+i\arccos(\lambda_i/α)}) + P(e^{-i\arccos(\lambda_i/α)})\right) δ_{ij} \\
&= \sum_{n=0}^d a_n \frac{1}{2}\left(e^{+i\arccos(λ_i/α)n} + e^{-i\arccos(λ_i/α)n}\right)δ_{ij} \\
&= p(\lambda_i/\alpha)δ_{ij}
\end{align}
as desired, with the identity $T_n(\cos x) = (e^{inx}+e^{-inx})/2$.


\emph{Possible polynomial transformations---}
In QSP and QET, polynomial transformations can be effected for any real, fixed parity polynomial satisfying the bound $\max_{x\in{-1,1}}|p(x)|\le1$. This is natural as singular values of any $p(A/α)$ cannot exceed one due to unitarity of its block encoding. Instead, GQET allows complex polynomials with indefinite parity but requires the bound \eqref{eqmain:gqsp polynomial} with the same coefficients as the transformation polynomial \eqref{eqmain:gqet polynomial} in the Chebyshev expansion. In the worst case, a polynomial $p(x)$ might therefore have to be scaled down by the scaling factor
\begin{equation}
\beta = \frac{\max_{|z|=1}|P(z)|}{\max_{x\in[-1,1]}|p(x)|} \le O(\log d)
\label{eqmain:scaling factor}
\end{equation}
in order to implement it. The logarithmic upped bound in degree $d$ can be shown using the periodic Hilbert transform, see appendix~\ref{sec:hilbert transform}.

In practice, this bound is loose and much better scaling factors are achieved. When $p(x)$ has only terms $n\equiv 1$ (or 3) mod 4, the bound $\beta\le2$ follows from $\max_{|z|=1}|P(z)| \le$
\begin{equation}
 \max_{γ\in[0,2π]}\left|\sum_{n=0}^d a_n \cos(nγ)\right| + \max_{γ\in[0,2π]}\left|\sum_{n=0}^d a_n i \sin(nγ)\right|
\end{equation}
$\le 2 \max_{x\in[-1,1]}|p(x)|$,
taking $γ\toγ+π/2$ (or $-π/2$) in the second term, and using $T_n(\cosγ)=\cos(nγ)$.

Further, we analyse the scaling factor numerically for the matrix inversion polynomial used in the quantum linear systems solver \cite{gilyenQuantumSingularValue2019,martynGrandUnificationQuantum2021}. We generated optimal degree polynomials ($d$ up to 2349) using QSPPACK's implementation of the Remez method \cite{dongEfficientPhasefactorEvaluation2021}. In all cases, both $|P(z)|$ and $|p(x)|$ are less than one (no scaling required), and the scaling factor is below 1.75. 

\emph{Singular Value Transformation (SVT)---}
While eigenvalue transformation is defined for Hermitian matrices, the singular value transformation is defined for all, including non-square, matrices $A\in\mathbb{C}^{N_L\times N_R}$ in terms of its singular value decomposition $A/α=W^\dag (D/α) V$, where $W,V$ are isometries, and $D\ge0$ square contains the singular values arranged on the diagonal. For odd and even parts of $p(x)$, the SVT is
\begin{align}
p_\text{odd}(A/α) &:= W^\dag p_\text{odd}(D/α) V,\nonumber\\
p_\text{even}(A/α) &:= V^\dag p_\text{even}(D/α) V. \label{eqmain:svt}
\end{align}
The split into transformations by odd and even parts of $p(x)$ is natural: The multiplications in an expression like $A+AA+AAA$ are ill-defined on dimensional grounds. Fixing this as $A+A^\dag A + AA^\dag A$, the additions are ill-defined on dimensional grounds unless the polynomial has definite parity.

\emph{Generalised Quantum Singular Value Transformation (GQSVT) via Hermitianisation---} We develop two methods lifting GQET to arbitrary matrices.
The first relies on a GQET circuit of the Hermitian $(N_L+N_R)\times (N_L+N_R)$ square matrix
\begin{equation}
    \bar A/\alpha = \begin{pmatrix} & A/\alpha \\ A^\dag/\alpha & \end{pmatrix}.
\end{equation}
A block encoding $\bar U\in\mathbb{C}^{2M\times2M}$ of the Hermitianised $\bar A$ is \cite{chakrabortyPowerBlockencodedMatrix2019}
\begin{equation}
\begin{tikzpicture}
\begin{yquant}
qubit {} old;
qubit {} new;
["north:$M$" {font=\protect\footnotesize, inner sep=0pt}]
slash old;
box {$\bar U$} (new, old);
text {$=$} (-);
slash old;
align -;
x new;
box {$U$} old ~ new;
box {$U^\dag$} old | new;
\end{yquant}
\end{tikzpicture}
\end{equation}
Insertion into the GQET circuit with phase factors for a polynomial $p(x)$ results in a block encoding of
\begin{equation}
    \begin{pmatrix} W^\dag p_\text{even}(D/α)W  & p_\text{odd}(A/α) \\ (p_\text{odd}(A/α))^\dag & p_\text{even}(A/α) \end{pmatrix}
    \label{eqmain:gqsvt hermitianisation result}
\end{equation}
from which the SVTs \eqref{eqmain:svt} can easily be extracted.
We delegate a derivation to appendix~\ref{sec:gqsvt via hermitianisation}, as it can most easily be shown for non-square matrices using projected unitary encodings. Note that GQSVT via Hermitianisation has double the query complexity to $U$ than QGET and QSVT. However, using a mixed parity polynomial, one can extract both even and odd SVTs from one circuit. Hermitianised block encodings also allow a derivation of QSVT from QET, where the extra qubit from Hermitianisation and the doubled query complexity of the block encoding can be simplified away, see appendix~\ref{sec:qsvt from qet}.

\emph{Generalised Quantum Singular Value Transformation (GQSVT) via multiplication---}
Block encoded matrices $A_1/α_1\in\mathbb{C}^{N_1\times N_2}$ and $A_2/α_2\in\mathbb{C}^{N_2\times N_3}$ (with block encodings $U_1,U_2\in U(M)$ respectively) may be multiplied. The circuit
\begin{equation}
\vcenter{\hbox{%
\begin{tikzpicture}
\begin{yquant}
    qubit {$\ket{0}$} extra;
    qubit {$\ket{0}$} flag;
    qubit {$\ket{j}$} block;
    ["north:$M/N_3$" {font=\protect\footnotesize, inner sep=0pt}]
    slash flag;
    ["north:$N_3$" {font=\protect\footnotesize, inner sep=0pt}]
    slash block;
    box {$U_2$} (flag, block);
    ["north:$M/N_2$" {font=\protect\footnotesize, inner sep=0pt}]
    slash flag;
    ["north:$N_2$" {font=\protect\footnotesize, inner sep=0pt}]
    slash block;
    cnot extra ~ flag;
    x extra;
    box {$U_1$} (flag, block);
    ["north:$M/N_1$" {font=\protect\footnotesize, inner sep=0pt}]
    slash flag;
    ["north:$N_1$" {font=\protect\footnotesize, inner sep=0pt}]
    slash block;
    output {$\ket{0}$} extra;
    output {$\ket{0}$} flag;
    output {$\ket{i}$} block;
\end{yquant}
\end{tikzpicture}}}%
= \frac{(A_1A_2)_{ij}}{\alpha_1\alpha_2}
\label{eqmain:multiplication}
\end{equation}
is a block encoding of their product, with only one extra qubit compared to the methods in \cite{gilyenQuantumSingularValue2019,vonburgQuantumComputingEnhanced2021}. (During preparation of this manuscript, \cite{dalzellQuantumAlgorithmsSurvey2023b} appeared with the same circuit \eqref{eqmain:multiplication}.) In fact, when there is no further processing of the block encoding, but just measurement, the multiplication can be performed without any extra qubits by a \emph{measure-early} strategy, consisting of mid-circuit measurement, postselection, and reinitialisation:
\begin{equation}
\begin{tikzpicture}
\begin{yquant}
qubit {$\ket{0}$} a;
qubit {} b;
slash b;
cnot a ~ b;
x a;
align -;
measure {$\ket{0}$} a;
discard a;
discard b;
text {$=$} (-);
settype {qubit} b;
slash b;
measure {$\ket{0}$} b ;
discard b;
init {\ reinit $\ket{0}$} b;
slash b;
\end{yquant}
\end{tikzpicture}
\end{equation}
Beyond saving a qubit, the advantage of measure-early multiplication is that circuit execution can be aborted and restarted early if the measurement fails on postselection (measurement result other than $\ket{0}$), saving time on the quantum computer.

\begin{figure*}
\begin{tikzpicture}
\begin{yquant}
qubit {$\ket{0}$} flag1;
qubit {$\ket{0}$} new;
qubit {$\ket{0}$} flags;
qubit {$\ket{j}$} block;
["north:$N$" {font=\protect\footnotesize, inner sep=0pt}]
slash block;
["north:$M/N$" {font=\protect\footnotesize, inner sep=0pt}]
slash flags;
%
box {$R(\theta_0, \phi'_0, λ)$} flag1;
hspace {2pt} -;
[name=u0]
box {$U$} (block, flags) ~ flag1;
cnot new ~ flags, flag1;
[name=x0]
x new ~ flag1;
[name=udag0]
box {$U^\dag$} (block, flags) ~ flag1;
[operator style={/yquant/every negative control}]
phase {} flag1 ~ flags, new;
box {$R(\theta_1, \phi_1', 0)$} flag1;
hspace {2pt} -;
[name=u1]
box {$U$} (block, flags) ~ flag1;
cnot new ~ flags, flag1;
[name=x1]
x new ~ flag1;
[name=udag1]
box {$U^\dag$} (block, flags) ~ flag1;
[operator style={/yquant/every negative control}]
phase {} flag1 ~ flags, new;
text {$\cdots$} -;
hspace {6pt} -;
\node[draw, dashed, fit=(u0) (x0) (udag0), label=below:{\protect\footnotesize block encoding of $A^\dag A/α^2$}] {};
\node[draw, dashed, fit=(u1) (x1) (udag1)] {};
\end{yquant}
\end{tikzpicture}

\begin{tikzpicture}
\begin{yquant}
qubit {$\hookrightarrow$} flag1;
qubit {$\hookrightarrow$} new;
qubit {$\hookrightarrow$} flags;
qubit {$\hookrightarrow$} block;
%
slash block;
slash flags;
box {$R(\theta_{\lfloor d/2\rfloor -1}, \phi_{\lfloor d/2 \rfloor -1}', 0)$} flag1;
hspace {2pt} -;
[name=u3]
box {$U$} (block, flags) ~ flag1;
cnot new ~ flags, flag1;
[name=x3]
x new ~ flag1;
[name=udag3]
box {$U^\dag$} (block, flags) ~ flag1;
[operator style={/yquant/every negative control}]
phase {} flag1 ~ flags, new;
box {$R(\theta_{\lfloor d/2\rfloor}, \phi_{\lfloor d/2 \rfloor}, 0)$} flag1;
align -;
output {$\ket{0}$} flag1, new;
[name=reinit]
measure {$\ket{0}$} flags;
discard flags;
init {$\ \text{reinit}\ \ket{0}$} flags;
[name=oddu]
box {$U$} (flags, block);
hspace {2pt} -;
output {$\ket{0}$} flags;
align -;
output {$\ket{i}$} block;
\node[draw, dashed, fit=(u3) (x3) (udag3)] {};
\node[draw, red, fit=(reinit) (oddu), label=below:{\protect\footnotesize odd case only}] {};
\end{yquant}
\end{tikzpicture}
\caption{Circuit for generalised quantum singular value transformation (GQSVT) via multiplication. This is a GQET (Fig.~\ref{figmain:gqet circuit}) with a block encoding of $A^\dag A/α^2$ (dashed rectangle). In the case of an odd polynomial, the red part of the circuit is required, in which the flag qubits for $U$ are measured, postselected as $\ket{0}$, and reused (measure-early multiplication).}\label{figmain:gqsvt via multiplication circuit}
\end{figure*}
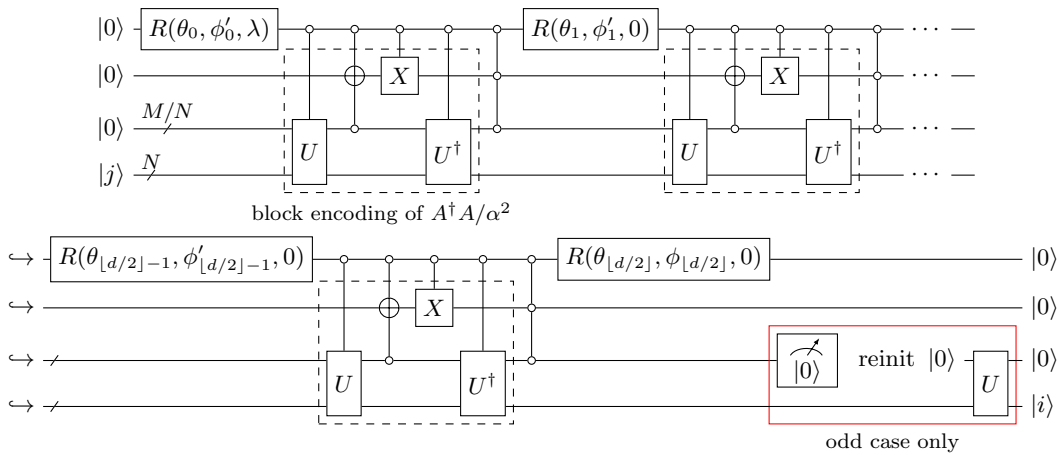

The GQSVT of $A/α$ by an even polynomial $p_\text{even}(x)$ of degree $d$ can be achieved by performing GQET of the Hermitian matrix $A^\dag A/α^2$ by
\begin{equation}
    q(x) = p_\text{even}(\sqrt{x})=\sum_{n=0}^{d/2} b_nT_n(x).\label{eqmain:q(x) even}
\end{equation}
Similarly, in the odd case, the GQET of $A^\dag A/α^2$ by 
\begin{equation}
    q(x) = p_\text{odd}(\sqrt{x})/\sqrt{x}=\sum_{n=0}^{(d-1)/2} b_nT_n(x)\label{eqmain:q(x) odd}
\end{equation}
must be left-multiplied by $A/α$ to retrieve the odd SVT of $A/α$.
The GQSVT circuit, Fig.~\ref{figmain:gqsvt via multiplication circuit}, is obtained by inserting the product \eqref{eqmain:multiplication} of $A_1 = A^\dag, A_2=A$ into the GQET circuit (Fig.~\ref{figmain:gqet circuit}). The query complexity is halved compared to GQSVT via Hermitianisation: Both $U$ and $U^\dag$ are present in each iteration, but the number of iterations is halved to $\lfloor d/2\rfloor$ due to the lower degree of $q(x)$.
In the odd case, the GQET circuit is multiplied by $A/α$; in Fig.~\ref{figmain:gqsvt via multiplication circuit} we show a measure-early multiplication.

In order for the GQET sequence to exist, $|Q(z)|=\left|\sum_{n=0}^{\lfloor d/2\rfloor}b_nz^n\right|$ must be bounded by 1 on $|z|=1$, with the coefficients of the Chebyshev expansion of $q(x)$ (\eqref{eqmain:q(x) even} or \eqref{eqmain:q(x) odd}). This can result in a different scaling factor than \eqref{eqmain:scaling factor} for GQET and GQSVT via Hermitianisation. Worries that the computation of the $b_n$ coefficients in \eqref{eqmain:q(x) even} or \eqref{eqmain:q(x) odd} (which are required to find the phase factors) from an expansion of $p(x)$ may be costly or numerically unstable can be eased: Rather than first finding a Chebyshev approximation $p(x)$ to a target function $f(x)$ (e.g., $f(x)\propto 1/x$ for the quantum linear systems solver), one can directly find a Chebyshev approximation $q(x)$ to $f(\sqrt{x})(/\sqrt{x})$.

\emph{Conclusions and Outlook---}
We have lifted GQSP to arbitrary matrices, introducing the GQET and two variants of GQSVT. For these, the polynomial $p(x)$ for the matrix transformation might have to be scaled down by a scaling factor \eqref{eqmain:scaling factor}. In some applications, like amplitude amplification, the polynomial must reach values very close to 1, and QET/QSVT will stay the algorithm of choice. In other yet unexplored contexts, when the polynomial is complex and/or (for Hermitian matrices) of mixed parity, GQET/GQSVT could outperform QET/QSVT, because it does not require linear combination of unitaries leading to downscaling.
In the future it would be desirable to understand better the tradeoff between any downscaling required in GQET/GQSVT and the downscaling required for complex and indefinite parity polynomials in QET/QSVT. 

Our presentation of Hermitianisation in terms of projected unitary encodings in appendix~\ref{sec:hermitianisation} may be of independent interest; Hermitianisation also allows a new derivation of QSVT from QET (appendix~\ref{sec:qsvt from qet}).

Further, we have shown how to perform multiplication with only one extra qubit  \eqref{eqmain:multiplication}, which is of independent interest. Additionally, \emph{measure-early} strategies as introduced here for multiplication should be considered whenever dealing with postselection to reduce average circuit length and QPU runtime.

\emph{Acknowledgements---} The author is grateful to Bjorn Berntson for the logarithmic bound on the scaling factor, and to Earl Campbell for helpful comments on the manuscript.

\makeatletter
\interlinepenalty=10000
\bibliography{main.bib}
\makeatother

\onecolumngrid
\clearpage
\tableofcontents

\appendix

\section{Overview of Appendices}
In the appendices, we give more details. It also expresses and shows everything in terms of projected unitary encodings, a generalisation of block encodings.

Figure~\ref{fig:overview} shows an overview of our results and methods used to obtain them.
First, in section~\ref{sec:encoding techniques} we define projected unitary encodings, block encodings, and present projected unitary encoding techniqes (qubitisation, Hermitianisation, multiplication) that are required.
Then, in section~\ref{sec:matrix trafos} we give a short overview of GQSP and show how the projected unitary encoding techniques can be used to derive the polynomial matrix transformations (GQET and both flavours of GQSVT) from GQSP.

Finally, appendix~\ref{sec:qsvt from qet} shows how to derive QSVT from QET, and appendix~\ref{sec:hilbert transform} shows how to derive the logarithmic bound of the scaling factor with the Hilbert transform.

\begin{figure}
\centering
\begin{tikzpicture}[xscale=1.9, yscale=4, every node/.style={font={\small}}]

\tikzstyle{arrow} = [thick,->,>=stealth,shorten >=2pt]
\tikzstyle{qsp} = [rectangle, rounded corners, align=center, draw=black]
\tikzstyle{be} = [rectangle, align=center, fill=riverlane_light_green!20]

\node (gqsp) at (-1.2, 0.1) [qsp] {\textbf{GQSP}\\  section~\ref{sec:gqsp}, \cite{motlaghGeneralizedQuantumSignal2023a}\\ Transformation by $P(z) = \sum_n a_n z^n$ \\ for unitary matrices};
\node (qubitisation) at (1.7, 0.1) [be] {\textbf{Qubitisation} \\ section~\ref{sec:qubitisation}, \cite{poulinQuantumAlgorithmSpectral2018,berryImprovedTechniquesPreparing2018}\\ Encoding $U$, reflection $\mathcal R_\Pi$\\ around projector\\ $\rightarrow$ qubitised unitary $\mathcal R_\Pi U$};
\node (gqet) at (0,-0.8) [qsp] {\textbf{GQET}\\ section~\ref{sec:gqet}\\ Transformation by $p(x) = \sum_n a_n T_n(x)$\\ for Hermitian matrices \\ with GQSP of qubitised unitary $\mathcal{R}_\Pi U$};
\node (hermitianise) at (-3, -1.1) [be] {\textbf{Hermitianisation}\\ section~\ref{sec:hermitianisation}, \cite{chakrabortyPowerBlockencodedMatrix2019} \\
Hermitian encoding of\\ $\bar A = \begin{pmatrix} & A/α \\ A^\dag/α & \end{pmatrix}$};
\node (multiply) at (3.1, -1.1) [be] {\textbf{Multiplication}\\ section~\ref{sec:multiplication}\\ Encoding of $A_1 A_2/(α_1α_2)$\\ with 1 extra qubit};
\node (gqsvtherm) at (-2.2, -1.8) [qsp] {\textbf{GQSVT via Hermitianisation}\\ section~\ref{sec:gqsvt via hermitianisation}\\ with GQET of Hermitianised $\bar A$};
\node (gqsvtmult) at (2.2,-1.8) [qsp] {\textbf{GQSVT via Multiplication}\\ section~\ref{sec:gqsvt via AA}\\ with GQET of product $A^\dag A$};
\draw [arrow] (gqsp) -- (gqet);
\draw [arrow] (qubitisation) -- (gqet);
\draw [arrow] (hermitianise) -- (gqsvtherm);
\draw [arrow] (gqet) -- node[anchor=west] {(*)} (gqsvtherm);
\draw [arrow] (gqet) -- (gqsvtmult);
\draw [arrow] (multiply) -- (gqsvtmult);
\end{tikzpicture}
\caption{Overview of results. Various
\protect\tikz [anchor=base, baseline]\protect\node [rectangle, fill=riverlane_light_green!20] {\small projected unitary encoding techniques}; (section~\ref{sec:encoding techniques}) are used to develop quantum circuits for  
\protect\tikz [anchor=base, baseline]\protect\node [rectangle, rounded corners, draw=black] {\small polynomial matrix transformations}; (section~\ref{sec:matrix trafos}) as a consequence of GQSP (generalised quantum signal processing \cite{motlaghGeneralizedQuantumSignal2023a}). While the generalised quantum eigenvalue transformation (GQET) can be used for Hermitian matrices with Hermitian projected unitary encodings, the generalised quantum singular value transformation (GQSVT) can be used for general, even non-square, matrices. Each of the two GQSVT methods can be advantageous in some situations. The method (*) to construct GQSVT from GQET using Hermitianised matrix encodings also offers a new derivation of the usual QSVT circuit from QET (appendix~\ref{sec:qsvt from qet}). 
} \label{fig:overview}
\end{figure}

\section{Projected unitary encoding techniques}
\label{sec:encoding techniques}
As quantum computation is unitary, non-unitary matrices must be encoded in larger unitary matrices to use them in a quantum circuit. In the main text this was achieved with block encodings \ref{eqmain:block encoding}. Here, we define the generalisation projected unitary encoding. They constitute the input and output model for the polynomial matrix transformations, as in \cite{gilyenQuantumSingularValue2019}.

\begin{definition}[Projected unitary encoding] Let $A$ be a complex $N_L\times N_R$ matrix, with a so-called subnormalisation $α\ge||A||_2$ (the spectral norm).
A projected unitary encoding $(U, \Pi_L, \Pi_R)$ consists of an $M\times M$ unitary $U$ ($M\ge N_L,N_R$), together with isometries $\Pi_L$ and $\Pi_R$ of dimensions $M\times N_L$ and $M\times N_R$, respectively.
They encode the matrix $A$ as
\begin{equation}
    \Pi_L^\dag U \Pi_R = A/α. \label{eq:projected unitary encoding}
\end{equation}
\end{definition}
Note that we have used rectangular isometries $\Pi_{L/R}$ in order to recover $A$'s dimensions in \eqref{eq:projected unitary encoding}. Instead, \cite{gilyenQuantumSingularValue2019} uses the $M\times M$ square projectors $\Pi_{L/R}\Pi_{L/R}^\dag$ to define projected unitary encodings. We can easily verify these are projectors by using the isometric property $\Pi_{L/R}^\dag\Pi_{L/R} =\mathbb{1}_{N_{L/R}\times N_{L/R}}$.
The condition $α\ge||A||_2$ for the subnormalisation is required for unitarity of $U$. In practice, $U$ must be decomposed into quantum gates, which will further increase the achievable subnormalisation.

A common scenario used in practice is a \emph{block encoding}, a special case of projected unitary encoding where $A/α$ is located in the top left block of $U$:
\begin{equation}
    U = \begin{pmatrix}A/\alpha & B_{N_L\times (M-N_R)} \\
    C_{(M-N_L)\times N_R} & D_{(M-N_L)\times (M-N_R)}\end{pmatrix},\ 
    \Pi_L = \begin{pmatrix} \mathbb{1}_{N_L\times N_L} \\ 0_{(M-N_L)\times N_L}\end{pmatrix},\ 
    \Pi_R = \begin{pmatrix} \mathbb{1}_{N_R\times N_R} \\ 0_{(M-N_R)\times N_R}\end{pmatrix}.
    \label{eq:block encoding}
\end{equation}

\begin{definition}[Hermitian projected unitary encoding]
When $A$ is a Hermitian square matrix ($N_L=N_R=N$), a projected unitary encoding $(U,\Pi_L,\Pi_R)$ is called Hermitian projected unitary encoding when $U$ is Hermitian and $\Pi_L = \Pi_R$.
\end{definition}
For a Hermitian block encoding \eqref{eq:block encoding}, $C = B^\dag$ follows. Note that even if a matrix $A$ of interest is Hermitian, one could have a quantum circuit implementing a non-Hermitian block encoding thereof. It is sensible to use circuit constructions that directly give a Hermitian block encoding without extra overhead, see e.g.~\cite{sunderhaufBlockencodingStructuredMatrices2023}.

\subsection{Qubitisation}
\label{sec:qubitisation}
A Hermitian projected unitary encoding $(U,\Pi)$ together with a reflection $\mathcal R_\Pi$ around the range of $\Pi$ leads to the qubitised walk operator \cite{poulinQuantumAlgorithmSpectral2018,berryImprovedTechniquesPreparing2018} 
\begin{equation}
    \mathcal R_\Pi U,\ \mathcal R_\Pi = -(\mathbb{1}-2\Pi\Pi^\dag).
\end{equation}
Each eigenvector $\vec\lambda_i$ of $A$ with eigenvalue $λ_i$ gives rise to a pair of eigenvalues $\Lambda_i^\pm$ and corresponding eigenvectors $\vec\Lambda^\pm$ of $U\mathcal R_\Pi$. Defining $\gamma_i=\arccos(\lambda_i/\alpha)\in[0,\pi]$, we write them down explicitly:
\begin{equation}
    \Lambda_i^\pm = e^{\pm i\gamma_i},\
    \vec\Lambda_i^\pm = \frac{1}{\sqrt{2}\sin\gamma_i}\left(e^{\pm i\gamma_i}\mathbb{1} - U\right)\Pi\vec\lambda_i\ \text{or for $λ_i/α=\pm1$,}\ \Lambda_i = \pm1,\ \vec\Lambda_i = \Pi\vec\lambda_i.
    \label{eq:qubitisation}
\end{equation}
In the special case that $U$ is a Hermitian block encoding, the eigenvectors simplify to
\begin{equation}
    \vec\Lambda_i^\pm = \frac{1}{\sqrt{2}}\begin{pmatrix}\pm i \vec \lambda_i\\\frac{B^\dag}{\sin \gamma_i}\vec\lambda_i\end{pmatrix}\ \text{or for $λ_i/α=\pm1$,}\ \vec\Lambda_i = \begin{pmatrix}\vec\lambda_i\\0\end{pmatrix}.
\end{equation}
These concise, normalised, and explicit expressions may be of independent interest when using and understanding qubitisation in other contexts.
Conceptually, Jordan's lemma is at work: $\mathcal R_\Pi$ and $U$ are direct sums of 2 dimensional reflections, such that their product is a direct sum of 2 dimensional rotations with eigenvalues $e^{\pm i\gamma_i}$. Equation~\eqref{eq:qubitisation} can be verified by direct application of $\mathcal R_\Pi U$:
\begin{align}
    -(\mathbb{1}-2\Pi\Pi^\dag)U\vec\Lambda_i^\pm
    &= - \frac{1}{\sqrt{2}\sin γ_i} (\mathbb{1}-2\Pi\Pi^\dag) U (e^{\pm i γ_i}\mathbb{1} - U) \Pi \vec\lambda_i \\
    &= \frac{1}{\sqrt{2}\sin{\gamma_i}} \left(-Ue^{\pm i\gamma_i}+U^2 + 2\Pi\Pi^\dag Ue^{\pm i\gamma_i} -2\Pi\Pi^\dag U^2\right) \Pi\vec\lambda_i\\
    &= \frac{1}{\sqrt{2}\sin{\gamma_i}} \big(-Ue^{\pm i\gamma_i} +\mathbb{1} +2e^{\pm i\gamma_i}\overbrace{\cosγ_i}^{\mathclap{\Pi^\dag U \Pi \vec\lambda_i = (A/α)\vec\lambda_i = \lambda_i/α\vec\lambda_i = \cos\gamma_i\vec\lambda_i}}\mathbb{1} - 2\mathbb{1} \big)\Pi \vec\lambda_i \\
    &= \frac{1}{\sqrt{2}\sin\gamma_i}\big(\underbrace{(2e^{\pm i\gamma_i}\cosγ_i -1)}_{\mathclap{=e^{\pm 2 i \gamma_i}}}\mathbb{1} - e^{\pm i\gamma_i}U\big)\Pi\vec\lambda_i \\
    &= e^{\pm i\gamma_i}\vec\Lambda^\pm_i.
\end{align}
In the case $\lambda_i/\alpha=\pm1\Leftrightarrow \gamma_i=0,\pi$, we only verify one eigenvector (as the other is contained in $\Pi$'s null space):
\begin{align}
    -(1-2\Pi\Pi^\dag)U\vec\Lambda_i =  (-U +2\Pi\Pi^\dag U)\Pi\vec\lambda_i  = -U\Pi\vec\lambda_i \pm 2\Pi\vec\lambda_i = \pm \Pi\vec\lambda_i,
\end{align}
where in the last step we have used that $|(-U \pm \mathbb{1})\Pi\vec\lambda_i|^2=0$ (the $\pm$ is whether $λ_i/α=\pm 1$). While $\mathcal{R}_\Pi U$ has more eigenvectors, those given in \eqref{eq:qubitisation} are sufficient to understand its action in the coding subspace (the range of $\Pi$) since
\begin{equation}
 \Pi\vec\lambda_i  = \frac{1}{\sqrt{2}i}\left(\vec\Lambda_i^+ - \vec\Lambda_i^-\right)
\end{equation}
or simply $\Pi\vec\lambda_i = \vec\Lambda_i$ if $\lambda_i/α=\pm1$.

Later we will require an anticontrolled application of the qubitised unitary in a quantum circuit. The (anti)controlled reflection $\mathcal{R}_\Pi$ can be implemented by a $\Pi\Pi^\dag$-controlled $-Z$ gate followed by a $-Z$. We indicate a projector-controlled gate by putting the projector into an ellipse, reminiscent of the circles used for $Z$ controls:
\begin{equation}
\begin{tikzpicture}
\begin{yquant}
qubit {} flag;
qubit {} block;
["north:$M$" {font=\protect\footnotesize, inner sep=0pt}]
slash block;
box {$R_\Pi U$} block ~ flag;
text {$=$} (-);
slash block;
box {$U$} block ~ flag;
box {$-(\mathbb{1}-2\Pi\Pi^\dag)$} block ~ flag;
text {$=$} (-);
slash block;
box {$U$} block ~ flag;
[name=proj, operator style={only at={1}{shape=yquant-circle}}]
box {\Ifnum\idx<1 $-Z$\Else $\Pi\Pi^\dag$\Fi} block, flag;
box {$-Z$} flag;
\draw (proj-0) -- (proj-1);
\end{yquant}
\end{tikzpicture}\label{eq:controlled reflection}
\end{equation}
In the special case of block encodings, this simplifies further:
\begin{equation}
\begin{tikzpicture}
\begin{yquant}
    qubit {} flag0;
    qubit {} flag1;
    qubit {} block;
    ["north:$N$" {font=\protect\footnotesize, inner sep=0pt}]
    slash block;
    ["north:$M/N$" {font=\protect\footnotesize, inner sep=0pt}]
    slash flag1;
    box {$\mathcal R_\Pi U$} (flag1, block) ~ flag0;
    text {$=$} (-);
    slash block;
    slash flag1;
    box {$U$} (flag1, block) ~ flag0;
    box {$\mathcal{R}_\Pi$} flag1 ~ flag0;
    text {$=$} (-);
    slash block;
    slash flag1;
    box {$U$} (flag1, block) ~ flag0;
    [name=proj, operator style={only at={1}{shape=yquant-circle}}]
    box {\Ifnum\idx<1 $-Z$\Else $\Pi\Pi^\dag$\Fi} flag0, flag1;
    box {$-Z$} flag0;
    text {$=$} (-);
    slash block;
    slash flag1;
    box {$U$} (flag1, block) ~ flag0;
    box {$-Z$} flag0 ~ flag1;
    box {$-Z$} flag0;
    \draw (proj-0) -- (proj-1);
\end{yquant}
\end{tikzpicture}
\label{eq:controlled reflection block encoding}
\end{equation}

\subsection{Hermitianisation}
\label{sec:hermitianisation}
For every matrix $A\in\mathbb{C}^{N_L\times N_R}$ and corresponding projected unitary encoding $(U,\Pi_L,\Pi_R)$, we can construct a Hermitian projected unitary encoding of the Hermitian $\bar N \times \bar N$ ($\bar N = N_L+N_R$) matrix \cite{chakrabortyPowerBlockencodedMatrix2019}
\begin{equation}
    \bar A/\alpha = \begin{pmatrix} & A/\alpha \\ A^\dag/\alpha & \end{pmatrix}
\end{equation}
with the same subnormalisation $α$
as follows. The Hermitian projected unitary encoding $(\bar U, \bar\Pi)$ has double the dimension $\bar M \times \bar M$  ($\bar M = 2M$) and is defined by the block matrix and projector
\begin{equation}
    \bar U = \begin{pmatrix} 0_{M\times M}& U \\ U^\dag & 0_{M\times M} \end{pmatrix},\ \bar\Pi = \begin{pmatrix}\Pi_L & 0_{M\times N_R}\\ 0_{M \times N_L}& \Pi_R\end{pmatrix}.
\end{equation}
The unitary $\bar U$ is manifestly Hermitian and encodes the correct matrix as can easily be verified by block matrix multiplication. We will need the reflection $\mathcal R_{\bar\Pi}$ to construct the qubitised operator, which is
\begin{equation}
    \mathcal R_{\bar\Pi} = -(\mathbb{1}-2\bar\Pi\bar\Pi^\dag) = \begin{pmatrix}\mathcal R_{\Pi_L} & 0_{M\times M} \\ 0_{M\times M} & \mathcal R_{\Pi_R}\end{pmatrix}
\end{equation}
Both $\bar U$ and $\mathcal R_{\bar\Pi}$ can be written in quantum circuit notation using
\begin{equation}
\bar U = (\ket{0}\bra{0}\otimes U + \ket{1}\bra{1}\otimes U^\dag)(X\otimes\mathbb{1}_{M\times M}),\ \mathcal{R}_{\bar\Pi} = \ket{0}\bra{0}\otimes\mathcal{R}_{\Pi_L} +\ket{1}\bra{1}\otimes\mathcal{R}_{\Pi_R}:
\end{equation}
\begin{equation}
\begin{tikzpicture}
\begin{yquant}
qubit {} new;
qubit {} old;
["north:$M$" {font=\protect\footnotesize, inner sep=0pt}]
slash old;
box {$\bar U$} (new, old);
text {$=$} (-);
slash old;
align -;
x new;
box {$U$} old ~ new;
box {$U^\dag$} old | new;
text {,$\qquad$} (-);
["north:$M$" {font=\protect\footnotesize, inner sep=0pt}]
slash old;
box {$\mathcal R_{\bar\Pi}$} (new, old);
text {$=$} (-);
slash old;
box {$\mathcal R_{\Pi_L}$} old ~ new;
box {$\mathcal R_{\Pi_R}$} old | new;
\end{yquant}
\end{tikzpicture}\label{eq:hermitianised circuits}
\end{equation}

\subsection{Multiplication}
\label{sec:multiplication}
Given two projected unitary encodings $(U_1, \Pi_{1,L}, \Pi_{1,R})$ of $A_1/α_1$ and $(U_2, \Pi_{2,L}, \Pi_{2,R})$ of $A_2/α_2$, we can construct a projected unitary encoding $(\bar U, \bar\Pi_L, \bar\Pi_R)$ of $A_1A_2/(α_1α_2)$. For this, the dimensions of $A_1$ ($N_1\times N_2$) and $A_2$ ($N_2\times N_3$) must be compatible. Further, we assume that both $U_1$ and $U_2$ have dimensions $M\times M$ (otherwise, the smaller one, w.l.o.g. $U_1$, can be easily padded eg.~as $(\mathbb{1}\otimes U_1, \ket{0}\otimes\Pi_{1,L},\ket{0}\otimes\Pi_{1,R})$). The encoding $\bar U$ of the product only uses one further qubit and is $2M\times 2M$:
\begin{align}
    &\bar U = \mathbb{1}_{2x2}\otimes U_1 \left(\mathbb{1}_{2\times 2}\otimes(\Pi_{2,R}\Pi_{1, L}^\dag) +  X \otimes(1-\Pi_{2,R}\Pi_{1,L}^\dag)\right)\mathbb{1}_{2x2}\otimes U_2,\\
    &\bar\Pi_L = \ket{0}\otimes\Pi_{2,L},\ \bar\Pi_R = \ket{0}\otimes\Pi_{1,R}
\end{align}
In quantum circuit form, this is:
\begin{equation}
\begin{tikzpicture}
\begin{yquant}
qubit {} flag;
qubit {} block;
["north:$M$" {font=\protect\footnotesize, inner sep=0pt}]
slash block;
box {$\bar U$} (flag, block);
text {$=$} (-);
slash block;
box {$U_2$} block;
[shape=yquant-circle, control style={only at={0}{/yquant/operators/every not}}]
box {$\Pi_{2,R}\Pi_{1,L}^\dag$} block | flag;
x flag;
box {$U_1$} block;
\end{yquant}
\end{tikzpicture}
\label{eq:multiplication circuit}
\end{equation}
The construction uses only one ancilla qubit, fewer than the constructions in \cite{gilyenQuantumSingularValue2019,lowHamiltonianSimulationQubitization2019}.  (During preparation of this manuscript, \cite{dalzellQuantumAlgorithmsSurvey2023b} appeared with the same circuit.)
When the projectors $\Pi_{1,L}=\Pi_{2,_R}, \Pi_{1,R}=\Pi_{2,L}$ (such as for the product $A^\dag A$ using projected encodings $(U,\Pi_L,\Pi_R)$ and $(U^\dag, \Pi_R,\Pi_L)$), the middle operation simply becomes a projector-controlled operation as in \eqref{eq:controlled reflection}, \eqref{eq:controlled reflection block encoding}. 

Let us demonstrate the multiplication circuit \eqref{eq:multiplication circuit} for $U_1,U_2$ block encodings. W.l.o.g.~say that $U_1\in U(M_1)$, $U_2\in U(M_2)$, $M_1\ge M_2$, such that $U_2$ must be padded:
\begin{equation}
\begin{tikzpicture}
\begin{yquant}
qubit {} extra;
qubit {} flags2;
qubit {} flags;
qubit {} block;
["north:$N_3$" {font=\protect\footnotesize, inner sep=0pt}]
slash block;
["north:$M_2/N_3$" {font=\protect\footnotesize, inner sep=0pt}]
slash flags;
["north:$M_1/M_2$" {font=\protect\footnotesize, inner sep=0pt}]
slash flags2;
box {$U_2$} (flags, block);
["north:$N_2$" {font=\protect\footnotesize, inner sep=0pt}]
slash block;
["north:$M_2/N_2$" {font=\protect\footnotesize, inner sep=0pt}]
slash flags;
cnot extra ~ flags;
x extra;
box {$U_1$} (block, flags, flags2);
["north:$N_1$" {font=\protect\footnotesize, inner sep=0pt}]
slash block;
["north:$M_2/N_1$" {font=\protect\footnotesize, inner sep=0pt}]
slash flags;
["north:$M_1/M_2$" {font=\protect\footnotesize, inner sep=0pt}]
slash flags2;
\end{yquant}
\end{tikzpicture} \label{eq:multiplication block encodings}
\end{equation}

When multiplication of block encodings (or projected unitary encodings) is the final step in a quantum circuit before measurement, the circuit can be simplified by a \emph{measure-early} approach. When all ancilla qubits in \eqref{eq:multiplication block encodings} (top three registers) are initialised as $\ket{0}$ and measured and postselected as $\ket{0}$, the circuit can be simplified. In fact, the top qubit can be removed, and we know the ccnot must have been triggered. This allows to introduce earlier measurements, postselection, and reinitialisation of qubits:
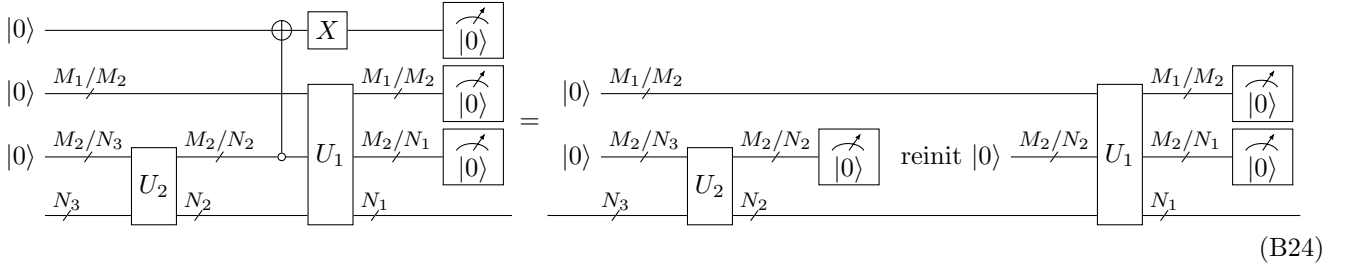
\begin{equation}
\begin{tikzpicture}
\begin{yquant}
qubit {$\ket{0}$} extra;
qubit {$\ket{0}$} flags2;
qubit {$\ket{0}$} flags;
qubit {} block;
["north:$N_3$" {font=\protect\footnotesize, inner sep=0pt}]
slash block;
["north:$M_2/N_3$" {font=\protect\footnotesize, inner sep=0pt}]
slash flags;
["north:$M_1/M_2$" {font=\protect\footnotesize, inner sep=0pt}]
slash flags2;
box {$U_2$} (flags, block);
["north:$N_2$" {font=\protect\footnotesize, inner sep=0pt}]
slash block;
["north:$M_2/N_2$" {font=\protect\footnotesize, inner sep=0pt}]
slash flags;
cnot extra ~  flags;
x extra;
box {$U_1$} (block, flags, flags2);
["north:$N_1$" {font=\protect\footnotesize, inner sep=0pt}]
slash block;
["north:$M_2/N_1$" {font=\protect\footnotesize, inner sep=0pt}]
slash flags;
["north:$M_1/M_2$" {font=\protect\footnotesize, inner sep=0pt}]
slash flags2;
measure {$\ket{0}$} extra, flags2, flags;
discard extra, flags2, flags;
text {$=$} (-);

init {$\ket{0}$} flags2;
init {$\ket{0}$} flags;
init {} block;
["north:$N_3$" {font=\protect\footnotesize, inner sep=0pt}]
slash block;
["north:$M_2/N_3$" {font=\protect\footnotesize, inner sep=0pt}]
slash flags;
box {$U_2$} (flags, block);
["north:$N_2$" {font=\protect\footnotesize, inner sep=0pt}]
slash block;
["north:$M_2/N_2$" {font=\protect\footnotesize, inner sep=0pt}]
slash flags;

measure {$\ket{0}$} flags;
discard flags;
init {\ reinit\ $\ket{0}$} flags;
["north:$M_1/M_2$" {font=\protect\footnotesize, inner sep=0pt}]
slash flags2;
["north:$M_2/N_2$" {font=\protect\footnotesize, inner sep=0pt}]
slash flags;
box {$U_1$} (block, flags, flags2);
["north:$N_1$" {font=\protect\footnotesize, inner sep=0pt}]
slash block;
["north:$M_2/N_1$" {font=\protect\footnotesize, inner sep=0pt}]
slash flags;
["north:$M_1/M_2$" {font=\protect\footnotesize, inner sep=0pt}]
slash flags2;
measure {$\ket{0}$} flags2, flags;
discard extra, flags2, flags;
\end{yquant}
\end{tikzpicture} \label{eq:multiplication measure early}
\end{equation}
The advantage of measure-early goes beyond the saving of one qubit (and freeing up the second $M_1/M_2$ register up for ancilla use in the first part of the circuit). Early measurement results in early failure; when the measurement does not yield $\ket{0}$ as desired. Then, the quantum computation can be aborted early and restarted; there is no need to execute the second part of the circuit. The total success probability for all postselections does not change by implementing a measure-early strategy, but the total circuit length run on the quantum computer is reduced.
Note that measure-early multiplication cannot be used if a block encoding for further processing is required, for example the result is fed into a QSVT.

\section{Polynomial matrix transformations}
\label{sec:matrix trafos}
\subsection{GQSP (generalised quantum signal processing) of unitary matrices}
\label{sec:gqsp}

Let $U$ be an $M\times M$ unitary matrix. Then GQSP \cite{motlaghGeneralizedQuantumSignal2023a} shows a quantum circuit construction to give a block encoding of $P(U)$, transformed by a (complex) degree-$d$ polynomial $P(z) =\sum_{n=0}^d a_n z^n$. As $U$ is unitary, $P(U)$ simply applies $P(z)$ to the eigenvalues. Such a GQSP transformation exists for all polynomials satisfying the bound $|P(z)|\le1\forall|z|=1$ on the unit circle. Then, phase factors $(\{\theta_i\}, \{\phi_i\}, λ)$ can be calculated in almost linear time $\tilde O(d)$ that specify the rotation
\begin{equation}
    R(\theta,\phi,\lambda) = \begin{pmatrix} e^{i(λ+\phi)}\cosθ & e^{i\phi}\sin θ\\ e^{iλ}\sinθ & -\cosθ\end{pmatrix}.
\end{equation}
Then the circuit GQSP$[U]=$
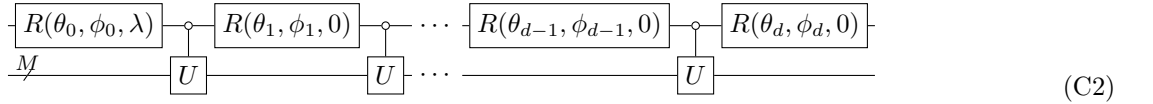
\begin{equation}
\label{eq:gqsp circuit}
\begin{tikzpicture}
\begin{yquant}
qubit {} flag1;
qubit {} block;
["north:$M$" {font=\protect\footnotesize, inner sep=0pt}]
slash block;
%
box {$R(\theta_0, \phi_0, λ)$} flag1;
box {$U$} block ~ flag1;
box {$R(\theta_1, \phi_1, 0)$} flag1;
box {$U$} block ~ flag1;
text {$\cdots$} -;
box {$R(\theta_{d-1}, \phi_{d-1}, 0)$} flag1;
box {$U$} block ~ flag1;
box {$R(\theta_{d}, \phi_{d}, 0)$} flag1;
\end{yquant}
\end{tikzpicture}
\end{equation}
block-encodes $P(U)=(\bra{0}\otimes\mathbb{1})\text{GQSP}[U](\ket{0}\otimes\mathbb{1})$. In particular, the isometry to extract $P(U)$ is $\Pi_L=\Pi_R=\ket{0}\otimes\mathbb{1}_{M\times M}$.

\subsection{GQET (generalised quantum eigenvalue transformation) of Hermitian matrices}
\label{sec:gqet}

The eigenvalue transformation of the Hermitian matrix $A/\alpha$ by the degree $d$ polynomial
\begin{equation}
    p(x) = \sum_{n=0}^d a_n T_n(x),\label{eq:p(x)}
\end{equation}
written in its expansion with Chebychev polynomials $T_n(x)$ is simply $p(A/\alpha)$, the polynomial applied to the eigenvalues.
We use the GQSP generalised phase factors $(\{\theta_i\},\{\phi_i\},\lambda)$ belonging to the GQSP polynomial
\begin{equation}
    P(z) =\sum_{n=0}^d a_n z^n\label{eq:P(z)}
\end{equation}
with the same coefficients as $p(x)$ in a monomial expansion.
Then a Hermitian projected encoding $(\text{GQET}[U], \ket{0}\otimes\Pi)$ of $p(A/\alpha)=\bra{0}\otimes\Pi^\dag \text{GQET}[U]\ket{0}\otimes\Pi$ is given by the GQSP circuit \eqref{eq:gqsp circuit} of the qubitised unitary \eqref{eq:controlled reflection}: $\text{GQET}[U] = \text{GQSP}[\mathcal R_\Pi U] = $
\begin{align}
\label{eq:gqet circuit}
\begin{tikzpicture}
\begin{yquant}
qubit {} flag1;
qubit {} block;
["north:$M$" {font=\protect\footnotesize, inner sep=0pt}]
slash block;
%
box {$R(\theta_0, \phi_0, λ)$} flag1;
box {$\mathcal{R}_\Pi U$} block ~ flag1;
box {$R(\theta_1, \phi_1, 0)$} flag1;
box {$\mathcal{R}_\Pi U$} block ~ flag1;
text {$\cdots$} -;
box {$R(\theta_{d-1}, \phi_{d-1}, 0)$} flag1;
box {$\mathcal{R}_\Pi$ U} block ~ flag1;
box {$R(\theta_{d}, \phi_{d}, 0)$} flag1;
\end{yquant}
\end{tikzpicture}
\end{align}
To see this, it is sufficient to consider the matrix elements of the encoded matrix with respect on the eigenvectors $\vec\lambda_i$ of $A$:
\begin{align}
    &\vec{\lambda_j}^\dag\left(\bra{0}\otimes\Pi^\dag \text{GQET}[U]\ket{0}\otimes\Pi\right) \vec\lambda_i 
    =\vec{\lambda_j}^\dag\Pi^\dag\left(\bra{0}\otimes\mathbb{1} \text{GQSP}[\mathcal R_\Pi U]\ket{0}\otimes\mathbb{1}\right)\frac{1}{\sqrt{2}i}\left(\Lambda^+_i -\Lambda^-_i\right) \\
    &\quad
    = \frac{1}{\sqrt{2}(-i)}\left((\vec\Lambda_j^+)^\dag -(\vec\Lambda_j^-)^\dag\right) \frac{1}{\sqrt{2}i} \left(P(e^{+i\gamma_i}) \vec\Lambda^+_i - P(e^{-i\gamma_i})\vec\Lambda^-_i\right) \\
    &\quad=\frac{1}{2}\left(P(e^{+i\gamma_i})+P(e^{-i\gamma_i})\right)δ_{ij} = p(\cos\gamma_i) δ_{ij} =p(λ_i/α)δ_{ij},
\end{align}
where we have used the fact that the GQSP circuit applies $P(z)$ to the eigenvalues of $\mathcal{R}_\Pi U$. We have also used the identity $T_n(x) = (e^{ixn}+e^{-ixn})/2$.

The quantum eigenvalue transformation (QET) based on QSP \cite{gilyenQuantumSingularValue2019} and GQET both require one ancilla qubit, and $d$ queries to $U$ (or controlled $U$). QET has the requirement that the polynomial $p(x)$ must be real, of definite parity, and satisfy $\max_{x\in[-1,1]}|p(x)| \le 1$. Polynomials with complex coefficients and indefinite parity require an LCU (linear combination of unitaries) construction in QET, which come with a qubit cost and cause subnormalisation of the result.

QGET in contrast directly allows polynomials with complex coefficients and indefinite parity. However, in GQET the bound on the polynomial is more confusing as $\max_{z\in\mathbb{C},|z|=1}|P(z)|\le1$ is required for a GQSP sequence to exist, where $P(z)$ in the monomial basis \eqref{eq:P(z)} has the same coefficients as the polynomial $p(x)$ effecting the transformation in the Chebyshev basis \eqref{eq:p(x)}.

Phase factors in GQSP can be computed more efficiently for a given polynomial, in almost linear time $\tilde O(d)$ \cite{motlaghGeneralizedQuantumSignal2023a}, while QSP phase factor algorithms require quadratic time $\tilde O(d^2)$ \cite{dongRobustIterativeMethod2023} in the degree of the polynomial.

\subsubsection{Possible polynomials}
We have the inequality
\begin{equation}
    \max_{x\in[-1,1]} | p(x) | = \max_{\gamma\in[0,2π]} \left|\sum_{n=0}^d a_n\frac{e^{iγn}+e^{-iγn}}{2}\right| \le \max_{|z=1|} |P(z)|
\end{equation}
using $x=\cos\gamma$ and the triangle inequality. This must also hold from unitarity; if $P(z)$ can be implemented by GQSP (i.e. RHS$\le1$), then $p(x)$ can be implemented by GQET.
Unfortunately the reverse inequality is not true. There are polynomials $\max_{x\in[-1,1]}|p(x)|\le1$ for which $\max_{|z|=1}|P(z)|>1$ and which therefore cannot be implemented in GQET.

Yet, every polynomial $p(x)/\beta$ can be implemented after scaling it down by the scaling factor
\begin{equation}
    \beta = \frac{\max_{|z|=1}|P(z)|}{\max_{x\in[-1,1]} |p(x)|}.
    \label{eq:scaling factor}
\end{equation}
Using the periodic Hilbert transform (see appendix~\ref{sec:hilbert transform}), one can prove an upper bound
\begin{equation}
\beta \le O(\log d)
\end{equation}
for the scaling factor that is only logarithmic in the degree $d$ of the polynomial.
Yet, we find that this upper bound does not become saturated in practice.
When $p(x)$ only has terms with $n\equiv 1 \pmod{4}$, the scaling factor is in fact bounded by 2:
\begin{align}
    \max_{|z|=1}|P(z)| &= \max_{\gamma\in[0,2π]}\left|\sum_{n=0}^d a_n\left[\frac{1}{2}\left(e^{inγ} + e^{-inγ}\right) + \frac{1}{2}\left(e^{inγ} - e^{-inγ}\right)\right]\right| \\
    &\le \max_{γ\in[0,2π]} | p(\cosγ)| + \max_{γ\in[0,2π]}\left|\sum_{n=0}^d a_n\frac{1}{2}\left( i e^{in(γ+π/2)} + i e^{-in(γ+π/2)}\right)\right| = 2\max_{x\in[-1,1]}|p(x)|.
\end{align}
Similarly, if $p(x)$ only has terms with $n\equiv3 \pmod{4}$, the scaling factor is also bounded by 2.

We further numerically analyse the scaling factor for a practically useful polynomial, the matrix inversion polynomial used in the quantum linear systems solver. With QSPPACK's implementation of the Remez algorithm \cite{dongEfficientPhasefactorEvaluation2021}, we generate several polynomials approximating
\begin{equation}
    p(x) \approx \frac{1}{4κ}\frac{1}{x},
\end{equation}
with a few choices of condition number $κ$ and accuracy $ε$ of the approximation. The results are collected in Table~\ref{table:scaling matrix inversion} and indicate scaling factors $\beta \le 1.75$ for a large range of degrees up to $d=2349$.

\begin{table}
\centering
\setlength{\tabcolsep}{10pt}
\begin{tabular}{ccc|ccc}
degree $d$ & $\kappa$ & $\epsilon$ & $\max_{x\in[-1,1]}|p(x)|$ & $\max_{|z|=1}|P(z)|$ & scaling factor $β$ \\
\hline
55 & 10 & 0.001 & 0.29 & 0.50 & 1.70 \\
79 & 10 & 0.0001 & 0.34 & 0.59 & 1.72 \\
221 & 40 & 0.001 & 0.29 & 0.50 & 1.70 \\
553 & 100 & 0.001 & 0.29 & 0.50 & 1.70 \\
783 & 100 & 0.0001 & 0.34 & 0.59 & 1.72 \\
1567 & 200 & 0.0001 & 0.34 & 0.59 & 1.72 \\
2349 & 300 & 0.0001 & 0.34 & 0.59 & 1.73 \\
\end{tabular}
\caption{Polynomials generated to approximate $\frac{1}{4\kappa}\frac{1}{x}$ to accuracy $ε$ with the Remez method implemented in QSPPACK \cite{dongEfficientPhasefactorEvaluation2021}. Data is shown for polynomials of optimal degree for various values of $\kappa$ and $\epsilon$. Data has been rounded to two digits.}\label{table:scaling matrix inversion}
\end{table}

\subsection{GQSVT (generalised quantum singular value transformation) of general matrices}
\label{sec:gqsvt}

Let $A$ be an arbitrary $N_L\times N_R$ matrix. It admits a singular value decomposition
\begin{equation}
    A = W^\dag D V,\qquad A\vec v_i = λ_i \vec w_i
\end{equation}
where $V$ is an $N_L\times N_L$ unitary, $W$ is an $N_R\times N_R$ unitary, and $D$ is $N_L\times N_R$ with the $\min(N_L, N_R)$ singular values arranged along the diagonal.
The singular value transformation by even $p_\text{even}(x)$ or odd $p_\text{odd}(x)$ polynomials are defined as:
\begin{align}
p_\text{odd}(A/α) := W^\dag p_\text{odd}(D/α) V,\quad p_\text{even}(A/α) := V^\dag p_\text{even}(D'/α) V, \label{eq:svt definition}
\end{align}
following \cite{gilyenQuantumSingularValue2019}.
For the even transformation, $D'$ is $D$ restricted to an $N_R\times N_R$ diagonal matrix, which keeps all of the singular values $\lambda_i$. (A similar definition and results for the even SVT keeping the left singular vectors is possible, but we use the right singular vectors for consistency with \cite{gilyenQuantumSingularValue2019}.)
In terms of right singular vectors $\vec v_i$ and left singular vectors $\vec w_i$ (columns of $V$ and $W$, respectively), it is:
\begin{equation}
    A\vec v_i = \lambda_i \vec w_i.
\end{equation}

The separation into even and odd parts is natural. In fact, for general matrices, the multiplications in an expression like $A+A^2+A^3$ are ill-defined simply on dimensional grounds. Instead, one can alternate $A$ and $A^\dag$ and use $A+A^\dag A +AA^\dag A$. Then, the additions are ill-defined on dimensional grounds, unless all terms are of fixed even or odd parity.

\subsubsection{GQSVT via Hermitianisation}
\label{sec:gqsvt via hermitianisation}

We will apply the GQET circuit to the Hermitianised projected encoding $(\bar U, \bar\Pi)$ from section~\ref{sec:hermitianisation}. To find $\text{GQET}[\bar U]$, the following eigenvalues and eigenvectors of $\bar A/α$ are useful:
\begin{equation}
    \begin{pmatrix}
        & A/α \\ A^\dag /α
    \end{pmatrix}
    \frac{1}{\sqrt{2}}
    \begin{pmatrix}
        \vec w_i \\
        \pm \vec v_i
    \end{pmatrix}
    = \pm \lambda_i 
    \frac{1}{\sqrt{2}}
    \begin{pmatrix}
        \vec w_i \\
        \pm \vec v_i
    \end{pmatrix}.
    \label{eq: Hermitianised eigenvalues}
\end{equation}

The GQSVT circuit for singular value transformation by $p_\text{odd}(A/α)$ or $p_\text{even}(A/α)$ is the GQET circuit \eqref{eq:gqet circuit} with a polynomial $p(x)$ having the desired odd $p_\text{odd}(x)$ or even $p_\text{even}(x)$ part for the Hermitianised projected unitary encoding \eqref{eq:hermitianised circuits}:
\begin{align}
&\text{GQSVT}[U] = \text{GQET}[\bar U]\nonumber\\
&\qquad
\begin{tikzpicture}
\begin{yquant}
qubit {} flag1;
qubit {} new;
qubit {} block;
["north:$M$" {font=\protect\footnotesize, inner sep=0pt}]
slash block;
align -;
%
box {$R(\theta_0, \phi_0, λ)$} flag1;
x new ~ flag1;
box {$\mathcal{R}_{\Pi_L} U$} block ~ new, flag1;
box {$\mathcal{R}_{\Pi_R} U^\dag$} block | new ~ flag1;
box {$R(\theta_1, \phi_1, 0)$} flag1;
x new ~ flag1;
box {$\mathcal{R}_{\Pi_L} U$} block ~ new, flag1;
box {$\mathcal{R}_{\Pi_R} U^\dag$} block | new ~ flag1;
hspace {10pt} -;
text {$\cdots$} -;
hspace {10pt} -;
\end{yquant}
\end{tikzpicture} \nonumber\\
&\qquad\quad
\begin{tikzpicture}
\begin{yquant}
qubit {$\hookrightarrow$} flag1;
qubit {$\hookrightarrow$} new;
qubit {$\hookrightarrow$} block;
box {$R(\theta_{d-1}, \phi_{d-1}, 0)$} flag1;
x new ~ flag1;
box {$\mathcal{R}_{\Pi_L} U$} block ~ new, flag1;
box {$\mathcal{R}_{\Pi_R} U^\dag$} block | new ~ flag1;
box {$R(\theta_{d}, \phi_{d}, 0)$} flag1;
\end{yquant}
\end{tikzpicture}\label{eq:gqsvt circuit}
\end{align}
For implementation details of the controlled reflections refer to \eqref{eq:controlled reflection} and \eqref{eq:controlled reflection block encoding}.
As mentioned in the main text \eqref{eqmain:gqsvt hermitianisation result}, the result is a projected unitary encoding $(\text{GQSVT}[U],\ket{0}\otimes\bar\Pi)$ of
\begin{equation}
    \bra{0}\otimes\bar\Pi^\dag\text{GQSVT}[U]\ket{0}\otimes\bar\Pi = \begin{pmatrix} W^\dag p_\text{even}(D/α)W  & p_\text{odd}(A/α) \\ (p_\text{odd}(A/α))^\dag & p_\text{even}(A/α) \end{pmatrix}.
    \label{eq:gqsvt hermitian result}
\end{equation}
This can be proven by considering the action on a basis
\begin{equation}
    \left\{ \begin{pmatrix}\vec v_i \\0 \end{pmatrix},\begin{pmatrix}0 \\ \vec w_i \end{pmatrix}\right\}
    \label{eq:gqsvt basis}
\end{equation}
composed of $A$'s singular vectors. We demonstrate the right column of \eqref{eq:gqsvt hermitian result}, the left column follows similarly. By expressing \eqref{eq:gqsvt basis} in terms of $\bar A$'s eigenvectors \eqref{eq: Hermitianised eigenvalues}, we can use the action of the GQET$[\bar U]$ circuit:
\begin{align}
    \bra{0}\otimes\bar\Pi^\dag \text{GQSVT}[U]\ket{0}\otimes\bar\Pi \begin{pmatrix}0 \\ \vec v_i\end{pmatrix} 
&= \bra{0}\otimes\bar\Pi^\dag \text{GQET}[\bar U]\ket{0}\otimes\bar\Pi \frac{1}{\sqrt{2}}\left(\frac{1}{\sqrt{2}}\begin{pmatrix}\vec w_i \\ \vec v_i \end{pmatrix} - \frac{1}{\sqrt{2}}\begin{pmatrix}\vec w_i \\ -\vec v_i \end{pmatrix}\right) \\
&=\frac{1}{2} \left(p(\lambda_i)\begin{pmatrix}\vec w_i \\\vec v_i\end{pmatrix} - p(-\lambda_i) \begin{pmatrix}\vec w_i \\-\vec v_i\end{pmatrix}\right) 
=\begin{pmatrix}p_\text{odd}(\lambda_i)\vec w_i \\ p_\text{even}(\lambda_i)\vec v_i\end{pmatrix},
\end{align}
which matches the definitions \eqref{eq:svt definition}.

If only $p_\text{odd}(A/α)$ or $p_\text{even}(A/α)$ are desired, they can be extracted from the GQSVT circuit by using isometries which extract the correct block from \eqref{eq:gqsvt hermitian result}. These give a prescription for initialisation and postselection of the top two qubits for projected encodings of odd or even SVT:
\begin{align}
    p_\text{odd}(A/α):&\quad \Pi_{\text{GQSVT},L} =  \ket{0}\otimes\bar\Pi\begin{pmatrix}\mathbb{1}_{N_L\times N_L}\\ 0 _{N_R\times N_L}\end{pmatrix} =  
    \ket{0}\otimes\ket{0}\otimes\Pi_L,\ \Pi_{\text{GQSVT},R} = \ket{0}\otimes\bar\Pi\begin{pmatrix}0_{N_L\times N_R} \\ \mathbb{1}_{N_R\times N_R}\end{pmatrix}= \ket{0}\otimes\ket{1}\otimes\Pi_R \label{eq:gqsvt odd isometries}\\
    p_\text{even}(A/α):&\quad \Pi_{\text{GQSVT},L} = \Pi_{\text{GQSVT},R} = \ket{0}\otimes\ket{1}\otimes\Pi_R,\label{eq:gqsvt even isometries}
\end{align}
where $\Pi_L, \Pi_R$ are simply the isometries of the original projected encoding $U$ of $A/α$.

Eigenvalue transformations of Hermitian square matrices $A$ can be performed with GQSVT when its projected unitary encoding $U$ is not unitary as required for GQET. The GQSVT then gives the eigenvalue transformations by $p_\text{odd}(x)$ and $p_\text{even}(x)$. One can also retrieve the eigenvalue transformation by the full $p(x)$ by using the isometry:
\begin{equation}
    \text{Hermitian $A$},\ p(A/α):\quad \Pi_{\text{GQSVT},L} =\Pi_{\text{GQSVT},R} = \frac{1}{\sqrt{2}}\ket{0}\otimes(\ket{0}\otimes\Pi_L + \ket{1}\otimes\Pi_R) = \frac{1}{\sqrt{2}}\ket{0}\otimes\begin{pmatrix}\Pi_L \\ \Pi_R\end{pmatrix}
\end{equation}

The method of establishing the GQSVT from the GQET using the Hermitianised projected encoding can also be used to derive the usual QSVT from the usual QET, see appendix~\ref{sec:qsvt from qet}.

\paragraph*{Comparison to QSVT}
The GQSVT circuit has one extra qubit and double the queries to controlled-$U$ and controlled-$U^\dag$ than QSVT to $U$ and $U^\dag$. From GQET it inherits the advantage of allowing complex polynomials without use of LCU, faster determination of phase factors, and the restrictions on possible polynomials.
The doubled query complexity can be seen as arising from the fact that a single run of the GQSVT$[U]$ circuit actually encodes both the SVTs by even and odd part of the polynomial, and either or even both can be extracted by judicious choice of isometries.

\subsubsection{GQSVT via multiplication}
\label{sec:gqsvt via AA}

The projected unitary encoding $\bar U$ of $A^\dag A/α^2$ is manifestly Hermitian (using section~\ref{sec:multiplication}).
For GQSVT$[U]$ by an even polynomial of degree $d$
\begin{equation}
    p_\text{even}(x) = \sum_{n=0, \text{even}}^{d/2} a_{2n} T_{2n}(x)
\end{equation}
we use $\text{GQET}[\bar U]$ with phase factors for the degree $d/2$ polynomial
\begin{equation}
    q(x) = \sum_{n=0,\text{even}}^{d/2} b_n T_n(x) := p_\text{even}(\sqrt{x}).
\end{equation}
Since $q(x)$ is applied to $A^\dag A/\alpha^2$, $(\text{GQSVT}[U] = \text{GQET}[\bar U], \ket{0}\otimes\bar\Pi_L, \ket{0}\otimes\bar\Pi_R)$ is a projected unitary encoding of $p_\text{even}(A/α)$.
Compared to the GQSVT via Hermitianisation (section~\ref{sec:gqsvt via hermitianisation}), we have half the query complexity to controlled-$U$ and controlled-$U^\dag$. Any necessary scaling of the polynomial could be different (better or worse), because the relevant condition from GQSP is $|Q(z)|=|\sum_n b_n z^n|\le 1\forall |z|=1$ with the coefficients $b_n$ of $q(x)$ rather than the $a_{2n}$ as in GQSVT via Hermitianisation.
Perhaps one might worry about the extra computational cost and numerical stability of finding the coefficients $b_n$ (from which the phase factors can be computed) from the coefficients $a_{2n}$. Yet, in practice, one does not start with a polynomial $p_\text{even}(x)$. Rather, one starts with a target function $f(x)$ and finds a polynomial approximation $p_\text{even}(x)\approx f(x)$. Then, for GQSVT via $A^\dag A$, one can instead directly find an approximation $q(x)\approx f(\sqrt{x})$. 

Similarly for odd polynomials of degree $d$, we use GQET of $A^\dag A$ with phase factors for the degree $(d-1)/2$ polynomial
\begin{equation}
    q(x) = p_\text{odd}(\sqrt{x})/\sqrt{x}.
\end{equation}
The result must be left-multiplied by $A/α$ to give a projected unitary encoding of $p_\text{odd}(A/α)$. At first it might seem like we need one more ancilla qubit than in the even case, due to the final multiplication by $A/α$. Yet, the ancilla qubit needed GQET can be reused after measuring it (and postselecting the result to be $\ket{0}$ as per the isometry). This is an application of measure-early multiplication (see section~\ref{sec:multiplication}. The GQSVT circuit is therefore:
\begin{align}
&\text{GQSVT}[U] = \nonumber\\
&\qquad
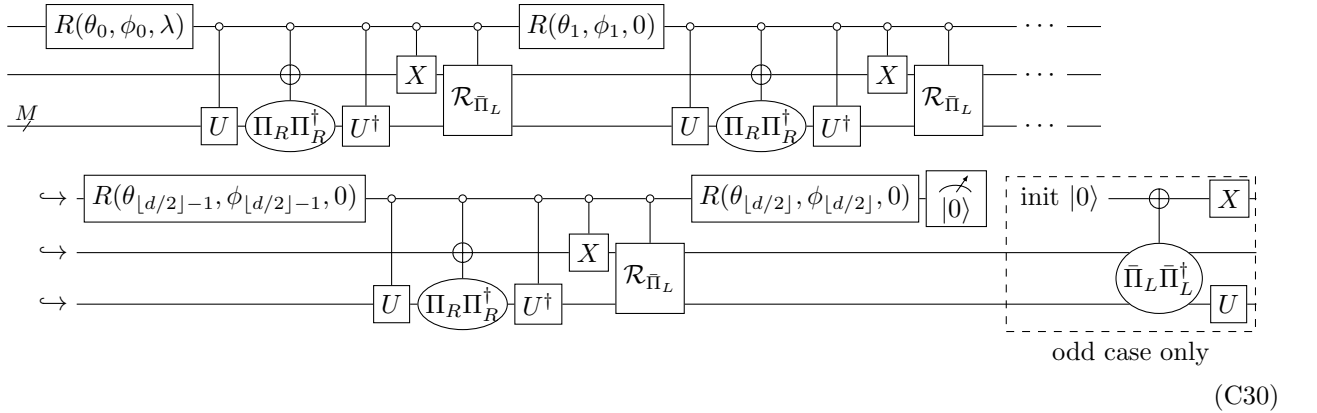
\begin{tikzpicture}
\begin{yquant}
qubit {} flag1;
qubit {} new;
qubit {} block;
["north:$M$" {font=\protect\footnotesize, inner sep=0pt}]
slash block;
align -;
%
box {$R(\theta_0, \phi_0, λ)$} flag1;
box {$U$} block ~ flag1;
[shape=yquant-circle, control style={only at={1}{/yquant/operators/every not}}]
box {$\Pi_R\Pi_R^\dag$} block ~ new, flag1;
box {$U^\dag$} block ~ flag1;
x new ~ flag1;
box {$\mathcal{R}_{\bar\Pi_L}$} (block,new) ~ flag1;
box {$R(\theta_1, \phi_1, 0)$} flag1;
box {$U$} block ~ flag1;
[shape=yquant-circle, control style={only at={1}{/yquant/operators/every not}}]
box {$\Pi_R\Pi_R^\dag$} block ~ new, flag1;
box {$U^\dag$} block ~ flag1;
x new ~ flag1;
box {$\mathcal{R}_{\bar\Pi_L}$} (block,new) ~ flag1;
hspace {10pt} -;
text {$\cdots$} -;
hspace {10pt} -;
\end{yquant}
\end{tikzpicture} \nonumber\\
&\qquad\quad
\begin{tikzpicture}
\begin{yquant}
qubit {$\hookrightarrow$} flag1;
qubit {$\hookrightarrow$} new;
qubit {$\hookrightarrow$} block;
box {$R(\theta_{\lfloor d/2 \rfloor-1}, \phi_{\lfloor d/2 \rfloor-1}, 0)$} flag1;
box {$U$} block ~ flag1;
[shape=yquant-circle, control style={only at={1}{/yquant/operators/every not}}]
box {$\Pi_R\Pi_R^\dag$} block ~ new, flag1;
box {$U^\dag$} block ~ flag1;
x new ~ flag1;
box {$\mathcal{R}_{\bar\Pi_L}$} (block,new) ~ flag1;
box {$R(\theta_{\lfloor d/2 \rfloor}, \phi_{\lfloor d/2 \rfloor}, 0)$} flag1;
measure {$\ket{0}$} flag1;
discard flag1;
hspace {8pt} flag1;
[name=init1]
init {init $\ket{0}$} flag1;
[shape=yquant-circle, control style={only at={0}{/yquant/operators/every not}}]
box {$\bar\Pi_L\bar\Pi_L^\dag$} (block, new) ~ flag1;
x flag1;
[name=u1]
box {$U$} block;
\node[draw, dashed, fit=(init1) (u1), label=below:{odd case only}] {};
\end{yquant}
\end{tikzpicture}\label{eq:gqsvt AA odd circuit}
\end{align}
with isometries
\begin{align}
    &\qquad \Pi_{\text{GQSVT},R} = \ket{0}\otimes\ket{0}\otimes\Pi_R, \\
    & \text{odd case:}\  \Pi_{\text{GQSVT}_L} = \ket{0}\otimes\ket{0}\Pi_L,\ \text{even case:}\ \Pi_{\text{GQSVT},L} = \Pi_{\text{GQSVT},R}.
\end{align}


\section{Derivation of QSVT from QET}
\label{sec:qsvt from qet}
The usual QSVT via QSP of a projected unitary encoding $U$ can be derived from the quantum eigenvalue transformation (QET) similarly to the construction of the GQSVT via Hermitianisation (section~\ref{sec:gqsvt via hermitianisation}).
First, a QET of the Hermitianised projected unitary encoding $\bar U$ is performed. The circuit requires $\bar\Pi\bar\Pi^\dag$-controlled nots, which can be split up into $\Pi_L\Pi_L^\dag$ and $\Pi_R\Pi_R^\dag$-controlled nots:
\begin{equation}
\begin{tikzpicture}
\begin{yquant}
qubit {} flag;
qubit {} new;
qubit {} block;
["north:$M$" {font=\protect\footnotesize, inner sep=0pt}]
slash block;
align -;

[plusctrl, shape=yquant-circle]
box {$\bar\Pi\bar\Pi^\dag$} (new, block) ~ flag;
text {$=$} (-);

slash block;
[shape=yquant-circle, control style={only at={0}{/yquant/operators/every not}}]
box {$\Pi_L\Pi_L^\dag$} block ~ new, flag;
[shape=yquant-circle, control style={only at={0}{/yquant/operators/every not}}]
box {$\Pi_R\Pi_R^\dag$} block | new, flag;
\end{yquant}
\end{tikzpicture}
\end{equation}
With this, the QET circuit for $\bar U$ is:
\begin{align}
&\begin{tikzpicture}
\begin{yquant}
qubit {} flag;
qubit {} new;
qubit {} block;
["north:$M$" {font=\protect\footnotesize, inner sep=0pt}]
slash block;
align -;
%
h flag;
[shape=yquant-circle, control style={only at={0}{/yquant/operators/every not}}]
box {$\Pi_L\Pi_L^\dag$} block ~ new, flag;
[shape=yquant-circle, control style={only at={0}{/yquant/operators/every not}}]
box {$\Pi_R\Pi_R^\dag$} block | new, flag;
box {$R_Z(\phi_1)$} flag;
[shape=yquant-circle, control style={only at={0}{/yquant/operators/every not}}]
box {$\Pi_L\Pi_L^\dag$} block ~ new, flag;
[shape=yquant-circle, control style={only at={0}{/yquant/operators/every not}}]
box {$\Pi_R\Pi_R^\dag$} block | new, flag;
x new;
box {$U$} block ~ new;
box {$U^\dag$} block | new;
hspace {10pt} -;
\end{yquant}
\end{tikzpicture} \nonumber\\
&\quad\begin{tikzpicture}
\begin{yquant}
qubit {$\hookrightarrow$} flag;
qubit {$\hookrightarrow$} new;
qubit {$\hookrightarrow$} block;
slash block;
align -;
%
[shape=yquant-circle, control style={only at={0}{/yquant/operators/every not}}]
box {$\Pi_L\Pi_L^\dag$} block ~ new, flag;
[shape=yquant-circle, control style={only at={0}{/yquant/operators/every not}}]
box {$\Pi_R\Pi_R^\dag$} block | new, flag;
box {$R_Z(\phi_2)$} flag;
[shape=yquant-circle, control style={only at={0}{/yquant/operators/every not}}]
box {$\Pi_L\Pi_L^\dag$} block ~ new, flag;
[shape=yquant-circle, control style={only at={0}{/yquant/operators/every not}}]
box {$\Pi_R\Pi_R^\dag$} block | new, flag;
x new;
box {$U$} block ~ new;
box {$U^\dag$} block | new;
hspace {10pt} -;
text {$\cdots$} (-);
hspace {10pt} -;
\end{yquant}
\end{tikzpicture}\nonumber\\
&\quad
\begin{tikzpicture}
\begin{yquant}
qubit {$\hookrightarrow$} flag;
qubit {$\hookrightarrow$} new;
qubit {$\hookrightarrow$} block;
slash block;
align -;
%
[shape=yquant-circle, control style={only at={0}{/yquant/operators/every not}}]
box {$\Pi_L\Pi_L^\dag$} block ~ new, flag;
[shape=yquant-circle, control style={only at={0}{/yquant/operators/every not}}]
box {$\Pi_R\Pi_R^\dag$} block | new, flag;
box {$R_Z(\phi_d)$} flag;
[shape=yquant-circle, control style={only at={0}{/yquant/operators/every not}}]
box {$\Pi_L\Pi_L^\dag$} block ~ new, flag;
[shape=yquant-circle, control style={only at={0}{/yquant/operators/every not}}]
box {$\Pi_R\Pi_R^\dag$} block | new, flag;
x new;
box {$U$} block ~ new;
box {$U^\dag$} block | new;
align -;
h flag;
\end{yquant}
\end{tikzpicture}
\label{eq:qet hermitianised circuit}
\end{align}
Crucially, the $X$ gate in the projected encoding $\bar U$ from \eqref{eq:hermitianised circuits} is not controlled, in contrast to GQSVT \eqref{eq:gqsvt circuit}. The second qubit, introduced by the Hermitianised encoding, is always initialised as $\ket{1}$ (see the right isometries in \eqref{eq:gqsvt odd isometries} and \eqref{eq:gqsvt even isometries}). Tracing the trajectory of the second qubit in the circuit, we see that it deterministically switches between $\ket{1}$ and $\ket{0}$. In fact, because of this, in each odd iteration only the $\Pi_R\Pi_R^\dag$-controlled not and $U$ are applied, and in each even iteration only the $\Pi_L\Pi_L^\dag$-controlled not and $U^\dag$ are applied. The second qubit can therefore be removed, and the usual QSVT circuit with alternating $U$ and $U^\dag$ is recovered.

\section{Periodic Hilbert transform}
\label{sec:hilbert transform}
This appendix proves the logarithmic bound $\beta \le O(\log d)$ for the scaling factor $\beta$ from \eqref{eq:scaling factor} by using the periodic Hilbert transform.
While the proof is directly applicable to polynomials $p(x)$ with real coefficients $a_n\in\mathbb{R}$, it can be used to bound the scaling factor of general polynomials by splitting $p(x) = p_1(x) + i p_2(x)$ into two polynomials with real parts and imaginary parts of the coefficients, respectively.
The author is grateful to Bjorn Berntson for the following proof, which is based on \cite[Theorem~15]{bauer2016}.

\subsection{Preliminaries and statement of results}

For $N\in \Z_{\geq 0}$, consider the polynomial 
\begin{equation}
p(z)=\sum_{n=0}^N a_n z^n \quad (z\in\C), 	
\end{equation}
where $\{a_n\}_{n=0}^N\subset \R$. We are interested in the behavior of this polynomial on the unit disk $\D\coloneqq\{z\in\C:|z|<1\}$. Suppose that the real part of this polynomial is bounded on $\partial \D$, i.e.,
\begin{equation}\label{eq:rep}
\big\lvert\re p(\ee^{\ii \theta})\big\rvert=\frac12\Bigg\lvert\sum_{n=0}^N a_n\big(\ee^{\ii n\theta}+\ee^{-\ii n\theta}\big)\Bigg\rvert \le M ,\qquad (\theta\in \T\coloneqq \R/ 2\pi\Z)
\end{equation}
for some $M>0$.
For real coefficients, $M$ is the denominator of the scaling factor \eqref{eq:scaling factor}, the maximum absolute value of the polynomial applied by GQET, because of the identity $(e^{\ii n\theta} + e^{-\ii n\theta})/2 = T_n(\cos\theta)$. The numerator of the scaling factor contains both the real and imaginary part, which we would like to bound in terms of the denominator $M$. 
Thus, we would like to use the assumption \eqref{eq:rep} to say something about bounds on the imaginary part of $p(z)$ on $\partial \D$, i.e.,
\begin{equation}\label{eq:imp}
\im p(\ee^{\ii \theta})=\frac1{2\ii}\sum_{n=1}^N a_n\big(\ee^{\ii n\theta}-\ee^{-\ii n\theta}\big) \quad (\theta\in\T).
\end{equation}
The (classical) fact that the (boundary values of the) real and imaginary parts of a holomorphic function are related by a Hilbert transformation suggests we should first answer the following (more general) question: for sufficiently regular functions on the torus ($\T$), how is the (periodic) Hilbert transform bounded in the norm $\lVert\cdot\rVert_{\infty,\T}$?
\\

Recall that the periodic Hilbert transform is defined as
\begin{equation}
(Hf)(\theta)\coloneqq\frac{1}{\pi}\pvint_{-\pi}^{\pi} \cot\bigg(\frac{\theta'-\theta}2\bigg)f(\theta')\,\mathrm{d}\theta'	,
\end{equation}
where the dashed integral denotes a principal value prescription with respect to the singularity of the integrand. For this operator, we have the well-known identity
\begin{equation}\label{eq:Hexp}
(H\ee^{\ii \omega \cdot})(\theta)=\ii\, \mathrm{sgn}(\omega) \ee^{\ii \omega\theta} \quad (\omega\in \R, \theta\in\T).
\end{equation}
which implies
\begin{equation}
(H\re p(\ee^{\ii\cdot}))(\theta)=-\im p(\ee^{\ii\theta}) \quad (\theta\in\T)
\end{equation}
and
\begin{equation}
 	\big\lvert H(\re p(\ee^{\ii\cdot}))(\theta)\big\rvert=\big\lvert \im p(\ee^{\ii\theta})\big\rvert \quad (\theta\in\T),
\end{equation}
and hence, uniform bounds on the Hilbert transform will give us uniform bounds on the imaginary part of our polynomial. Our main result is the following, a refinement of an inequality in \cite[Theorem~15]{bauer2016}.  

\begin{theorem}\label{thm:main}
For $f\in C^{\infty}(\T;\C)$ and $\delta\in(0,1]$, the following inequality holds
\begin{align}\label{eq:mainestimate}
\lVert Hf\rVert_{\infty,\T}\leq &\; g_1\frac{4}{\pi}\bigg\lvert\log \bigg(\frac{\delta}{2}\bigg)\bigg\rvert \lVert f\rVert_{\infty,\T} +g_1 \frac{\sqrt{2\delta}}{\pi} \sqrt{2-2\log\bigg(\frac{\delta}{2}\bigg)+\log^2\bigg(\frac{\delta}{2}\bigg)}\lVert f'\rVert_{2,\T}
\end{align}
with $g_1\coloneqq \log(\sin(\frac12))/\log(\frac12)$. 
\end{theorem}

Theorem~\ref{thm:main} is proved in Section~\ref{sec:thmproof}. To employ Theorem~\ref{thm:main} to bound polynomials, we recall the following fundamental result of Bernstein for bounding trigonometric (Laurent) polynomials.

\begin{theorem}[Bernstein \cite{bernstein1912}]
\label{thm:bernstein}
For $\{a_n\}_{n=-N}^N\subset \C$, let
\begin{equation}
f(\theta)=\sum_{n=-N}^N a_n \ee^{\ii n\theta	} \quad (\theta\in \T).
\end{equation}
Then,
\begin{equation}
\lVert f'\rVert_{\infty,\T}\leq N \lVert f\rVert_{\infty,\T}.	
\end{equation}
\end{theorem}

Combining Theorems~\ref{cor:main} and \ref{thm:bernstein} yields the following corollary, which is proved in Section~\ref{sec:corproof}.

\begin{corollary}\label{cor:main}
For $p\in \R[z]$ with $\deg p=N$ ($N\in \Z_{\geq 1}$) satisfying \eqref{eq:rep}, the following inequality holds
\begin{equation}\label{eq:corestimate}
\lvert p(z) \rvert \le M\bigg(1+\frac{g_1}{\pi}\bigg( 4\log\big(2N^2\big)	+\sqrt{4+4\log\big(2N^2\big)+2\log^2\big(2N^2\big)}\bigg)\bigg) \quad (z\in \partial \D).
\end{equation}	
\end{corollary}

This result is straightforwardly extended to complex polynomials in the following corollary, which is proved in Section~\ref{subsec:complexproof}.
\begin{corollary}\label{cor:complex}
For $p\in \C[z]$ with $\deg p=N$ ($N\in \Z_{\geq 1}$) satisfying $|\re p(z)|<M$ on $\partial \D$, the following inequality holds
\begin{equation}\label{eq:corestimate2}
\lvert p(z) \rvert \le M\bigg(1+\lvert\im p(0)\rvert+\frac{g_1}{\pi}\bigg( 4\log\big(2N^2\big)	+\sqrt{4+4\log\big(2N^2\big)+2\log^2\big(2N^2\big)}\bigg)\bigg) \quad (z\in \partial \D).
\end{equation}	
\end{corollary}

\begin{remark}
By straightforward analysis of the function
\begin{equation}
h(x)\coloneqq \sqrt{4+4x+2x^2},
\end{equation}
we obtain the inequality
\begin{equation}
	\sqrt{4+4x+2x^2}\leq h(\log(2))+\sqrt{2}(x-\log(2))
\end{equation}
on the interval $[\log(2),\infty)$, which we can use to simplify the bounds \eqref{eq:corestimate} and \eqref{eq:corestimate2} (at the expense of increasing them) to
\begin{equation}\label{eq:corestimatealt}
\lvert p(z) \rvert \le M\bigg(1+\frac{g_1}{\pi}\bigg(h(\log(2))-\sqrt{2}\log(2)+(4+\sqrt{2})\log\big(2N^2\big)	\bigg)\bigg) \quad (z\in \partial \D)
\end{equation}	
and
\begin{equation}\label{eq:corestimate2alt}
\lvert p(z) \rvert \le M\bigg(1+\lvert\im p(0)\rvert+\frac{g_1}{\pi}\bigg(h(\log(2))-\sqrt{2}\log(2)+(4+\sqrt{2})\log\big(2N^2\big)	\bigg)\bigg) \quad (z\in \partial \D),
\end{equation}	
respectively.
\end{remark}

\subsection{Proofs}
\subsubsection{Proof of Theorem~\ref{thm:main}}
\label{sec:thmproof}
By the definitions of the uniform norm and the periodic Hilbert transform and periodicity, we write
\begin{align}\label{eq:Hfnorm}
\lVert Hf\rVert_{\infty,\T}=\sup_{\theta\in \T}\frac1{2\pi}\Bigg\lvert\,\,\pvint_{-\pi}^{\pi} \cot\bigg(\frac{\theta'}{2}\bigg)f(\theta'-\theta)\,\mathrm{d}\theta'\Bigg\rvert. 	
\end{align}
We partition the domain of integration as
\begin{align}\label{eq:threeints}
\frac1{2\pi}\pvint_{-\pi}^{\pi} \cot\bigg(\frac{\theta'}{2}\bigg)f(\theta'-\theta)\,\mathrm{d}\theta'=\frac{1}{2\pi}\Bigg(\,\,\pvint_{-\delta}^{\delta}+\int_{-\pi}^{-\delta}+\int_{\delta}^{\pi}\Bigg)\cot\bigg(\frac{\theta'}{2}\bigg)f(\theta'-\theta)\,\mathrm{d}\theta'.
\end{align}
The first term of \eqref{eq:threeints} is integrated by parts to obtain
\begin{align}\label{eq:byparts}
\frac{1}{2\pi}\pvint_{-\delta}^{\delta} \cot\bigg(\frac{\theta'}{2}\bigg)f(\theta'-\theta)\,\mathrm{d}\theta'= &\; \frac{1}{2\pi}\lim_{\epsilon\downarrow 0} \Bigg(\int_{-\delta}^{-\epsilon}+\int_{\epsilon}^{\delta}\Bigg)	\cot\bigg(\frac{\theta'}{2}\bigg)f(\theta'-\theta)\,\mathrm{d}\theta' \nonumber\\
=&\;\lim_{\epsilon\downarrow 0}\Bigg( \frac{1}{\pi}\bigg[\log\bigg\lvert \sin\bigg(\frac{\theta'}{2}\bigg)\bigg\rvert f(\theta'-\theta)\bigg]_{\theta'=-\delta}^{\theta'=-\epsilon}+\frac{1}{\pi}\bigg[\log\bigg\lvert \sin\bigg(\frac{\theta'}{2}\bigg)\bigg\rvert f(\theta'-\theta)\bigg]_{\theta'=\epsilon}^{\theta'=\delta}\Bigg) \nonumber \\
&\; -\frac{1}{\pi}\lim_{\epsilon\downarrow 0} \Bigg(\int_{-\delta}^{-\epsilon}+\int_{\epsilon}^{\delta}\Bigg)	\log\bigg\lvert\sin\bigg(\frac{\theta'}{2}\bigg)\bigg\rvert f'(\theta'-\theta)  \,\mathrm{d}\theta' \nonumber\\
= &\; \frac{1}{\pi}\bigg\lvert\log\bigg( \sin\bigg(\frac{\delta}{2}\bigg)\bigg)\bigg\rvert \big(f(\delta-\theta)-f(-\delta-\theta)\big)-\frac{1}{\pi}\int_{-\delta}^{\delta}\log\bigg\lvert\sin\bigg(\frac{\theta'}{2}\bigg)\bigg\rvert f'(\theta'-\theta)\,\mathrm{d}\theta',
\end{align}
where, in the third step, we have employed the limit (recall that $f$ is smooth)
\begin{equation}
\lim_{\epsilon\downarrow 0} \big(f(\epsilon-\theta)-f(-\epsilon-\theta)\big)\log\bigg(\sin\bigg(\frac{\epsilon}{2}\bigg)\bigg)=0
\end{equation}
and used the fact that the singularity of $\log\big\lvert\sin\big(\frac{\theta'}{2}\big)\big\rvert$ is integrable to remove the principal value prescription. \\

Together \eqref{eq:Hfnorm}--\eqref{eq:byparts} imply 
\begin{align}\label{eq:threeterms}
\lVert Hf\rVert_{\infty,\T}\leq &\;  \frac{1}{\pi}\bigg\lvert\log\bigg( \sin\bigg(\frac{\delta}{2}\bigg)\bigg)\bigg\rvert \big(\lvert f(\delta-\theta)\rvert+\lvert f(-\delta-\theta)\rvert \big)+\frac1{\pi}\int_{-\delta}^{\delta}\log\bigg\lvert\sin\bigg(\frac{\theta'}{2}\bigg)\bigg\rvert \lvert f'(\theta'-\theta)\rvert \,\mathrm{d}\theta' \nonumber \\
&\; +\frac1{2\pi}\Bigg(\int_{-\pi}^{-\delta}+\int_{\delta}^{\pi}\Bigg)\bigg\lvert \cot\bigg(\frac{\theta'}{2}\bigg)\bigg\rvert \lvert f(\theta'-\theta)\rvert \,\mathrm{d}\theta'.
\end{align}
Before proceeding, we establish the following technical lemma. 
\begin{lemma}\label{lem:technical}
The function
\begin{equation}\label{eq:gdefinition}
g(x)\coloneqq \frac{\log\big(\sin\big(\frac{x}2\big)\big)}{\log\big(\frac{x}{2}\big)}
\end{equation}
satisfies 
\begin{equation}
\sup\limits_{x\in(0,1]} \lvert g(x)\rvert=g_1\coloneqq g(1)\approx	 1.06.
\end{equation}
\end{lemma}

\begin{proof}
First note that, by L'H\^{o}pital's rule,
\begin{equation}
\lim_{x\downarrow 0} g(x)=1. 	
\end{equation}
The derivative of \eqref{eq:gdefinition} is
\begin{equation}
g'(x)=\frac{\frac12\cot\big(\frac{x}2\big)\log\big(\frac{x}2\big)-\frac1x\log\big(\sin\big(\frac{x}2\big))}{\log\big(\frac{x}2\big)^2}.
\end{equation}
We claim that $g'(x)>0$ (and thus $g(x)>0$) on $(0,1)$. The Laurent series \cite[Eq.~4.19.6]{DLMF}
\begin{equation}
\frac12\cot\bigg(\frac{z}{2}\bigg)=\frac1z-\sum_{n=1}^{\infty}\frac{(-1)^{n+1}2^{2n-1}B_{2n}}{(2n)!}	\bigg(\frac{z}{2}\bigg)^{2n-1},
\end{equation}
where the $B_{2n}$ are the Bernoulli numbers \cite[Chapter~24.2]{DLMF}, coverges on $\overline{\D}\setminus\{0\}\supset(0,1]$. This implies that 
\begin{equation}
\frac1x-\frac12\cot\bigg(\frac{x}{2}\bigg)=\sum_{n=1}^{\infty}\frac{(-1)^{n+1}2^{2n-1}B_{2n}}{(2n)!}	\bigg(\frac{x}{2}\bigg)^{2n-1},
\end{equation}
on $(0,1]$. Note that for each $n\in \Z_{\geq 1}$, the summand is positive on $(0,1]$ owing to the alternating sign property of the Bernoulli numbers: $\mathrm{sgn}(B_{2n})=(-1)^{n+1}$ for $n\in \Z_{\geq 1}$; it follows that
\begin{equation}\label{eq:inequality1}
\frac12\cot\bigg(\frac{x}{2}\bigg)< \frac1x \quad (x\in(0,1]). 
\end{equation}
Moreover, using the inequality
\begin{equation}
\sin(x)<x \quad (x\in \R_{> 0})
\end{equation}
and the fact that $\log(x)$ is non-positive and increases monotonically on $(0,1]$, we deduce
\begin{equation}\label{eq:inequality2}
\log\bigg(\sin\bigg(\frac{x}{2}\bigg)\bigg)<\log\bigg(\frac{x}{2}\bigg)	\leq 0 \quad (x\in (0,1]). 
\end{equation}
Together, \eqref{eq:inequality1} and \eqref{eq:inequality2} imply $g'(x)>0$ on $(0,1]$. The result follows. 
\end{proof}

We bound each term in \eqref{eq:threeterms}, using Lemma~\ref{lem:technical} as appropriate. For the first term,
\begin{align}\label{eq:firstterm}
	\frac{1}{\pi}\bigg\lvert\log\bigg( \sin\bigg(\frac{\delta}{2}\bigg)\bigg)\bigg\rvert \big(\lvert f(\delta-x)\rvert+\lvert f(-\delta-x)\rvert \big)\leq \frac{2}{\pi}\bigg\lvert\log\bigg( \sin\bigg(\frac{\delta}{2}\bigg)\bigg)\bigg\rvert\lVert f\rVert_{\infty,\T}< g_1 \frac{2}{\pi}\bigg\lvert\log\bigg(\frac{\delta}{2}\bigg) \bigg\rvert \lVert f\rVert_{\infty,\T}.
\end{align}

For the second term in \eqref{eq:threeterms}, we first use the Cauchy-Schwarz inequality to write
\begin{align}\label{eq:cauchyschwarz}
\frac{1}{\pi}\int_{-\delta}^{\delta} 	\log\bigg\lvert\sin\bigg(\frac{x'}{2}\bigg)\bigg\rvert \big\lvert f'(x'-x)\big\rvert\,\mathrm{d}x' \leq &\; 
\frac1{\pi}\Bigg(\int_{-\delta}^{\delta} \log^2\bigg\lvert\sin\bigg(\frac{x'}{2}\bigg)\bigg\rvert\,\mathrm{d}x' \Bigg)^{\frac12}\Bigg(\int_{-\delta}^{\delta} \big\lvert f(x'-x)\big\rvert^2 \,\mathrm{d}x' \Bigg)^{\frac12} \nonumber \\
 \leq &\; \frac1{\pi}\Bigg(\int_{-\delta}^{\delta} \log^2\bigg\lvert\sin\bigg(\frac{x'}{2}\bigg)\bigg\rvert\,\mathrm{d}x' \Bigg)^{\frac12} \lVert f'\rVert_{2,\T}.
\end{align}
The integral in the second line of \eqref{eq:cauchyschwarz} can be estimated as
\begin{align}
\int_{-\delta}^{\delta} \log^2\bigg\lvert\sin\bigg(\frac{x'}{2}\bigg)\bigg\rvert\,\mathrm{d}x' \leq  2g_1^2 \int_{0}^{\delta} \log^2\bigg(\frac{x'}{2}\bigg)\,\mathrm{d}x'=2g_1^2\delta\Bigg(2-2\log\bigg(\frac{\delta}{2}\bigg)+\log^2\bigg(\frac{\delta}{2}\bigg)\Bigg),
\end{align}
giving 
\begin{align}\label{eq:secondterm}
\frac{1}{\pi}\int_{-\delta}^{\delta} 	\log\bigg\lvert\sin\bigg(\frac{x'}{2}\bigg)\bigg\rvert \big\lvert f'(x'-x)\big\rvert\,\mathrm{d}x'\leq g_1\frac{\sqrt{2\delta}}{\pi}\sqrt{2-2\log\bigg(\frac{\delta}{2}\bigg)+\log^2\bigg(\frac{\delta}{2}\bigg)} \lVert f'\rVert_{2,\T}.
\end{align}

The third term in \eqref{eq:threeterms} is bounded as
\begin{align}\label{eq:thirdterm}
\frac1{2\pi}\Bigg(\int_{-\pi}^{-\delta}+\int_{\delta}^{\pi}\Bigg)\bigg\lvert \cot\bigg(\frac{x'}{2}\bigg)\bigg\rvert \lvert f(x'-x)\rvert \,\mathrm{d}x' \leq &\; \frac{1}{\pi}\lVert f\rVert_{\infty,\T} \int_{\delta}^{\pi} \cot\bigg(\frac{x'}{2}\bigg)\,\mathrm{d}x' \nonumber\\
=&\; \frac{2}{\pi}\lVert f\rVert_{\infty,\T}\bigg\lvert \log\bigg(\sin\bigg(\frac{\delta}{2}\bigg)\bigg)\bigg\rvert	\nonumber \\
<&\; g_1 \frac{2}{\pi}\bigg\lvert\log\bigg(\frac{\delta}{2}\bigg) \bigg\rvert \lVert f\rVert_{\infty,\T}.
\end{align}
Putting \eqref{eq:firstterm}, \eqref{eq:secondterm}, and \eqref{eq:thirdterm} in \eqref{eq:threeterms} gives the result \eqref{eq:mainestimate}.

\subsubsection{Proof of Corollary~\ref{cor:main}}
\label{sec:corproof}
From the standard estimate
\begin{equation}
\lVert f\rVert_{2,\T}\leq \lVert f\rVert_{\infty,\T}	
\end{equation}
we have, under the assumptions of Theorem~\ref{thm:main},
\begin{align}\label{eq:mainestimatealt}
\lVert Hf\rVert_{\infty,\T}\leq &\; g_1\frac{4}{\pi}\bigg\lvert\log \bigg(\frac{\delta}{2}\bigg)\bigg\rvert \lVert f\rVert_{\infty,\T} +g_1 \frac{\sqrt{2\delta}}{\pi} \sqrt{2-2\log\bigg(\frac{\delta}{2}\bigg)+\log^2\bigg(\frac{\delta}{2}\bigg)}\lVert f'\rVert_{\infty,\T}.
\end{align}

Set $f(\theta)=\re p(\ee^{\ii\theta})$. It follows that 
\begin{equation}
\lVert f\rVert_{\infty,\T}\leq N,\quad \lVert f'\rVert_{\infty,\T}<MN,	
\end{equation}
where we have used Theorem~\ref{thm:bernstein}. We consider two cases.

\paragraph{Case 1: $\lVert f'\rVert_{\infty,\T}\leq \lVert f\rVert_{\infty,\T}$.}

Set $\delta=1$. Then, 
\begin{align}\label{eq:mainestimatealt1}
\lVert Hf\rVert_{\infty,\T}\leq &\; M\frac{g_1}{\pi}\big(4\log (2)+\sqrt{4+4\log(2)+2\log^2(2)}\big).
\end{align}

\paragraph{Case 2: $\lVert f'\rVert_{\infty,\T}> \lVert f\rVert_{\infty,\T}$.} Set $\delta=N^{-2}$. Then, 
\begin{align}\label{eq:mainestimatealt2}
\lVert Hf\rVert_{\infty,\T}\leq &\; M\frac{g_1}{\pi}\bigg(4\log\big(2N^2\big)+\sqrt{4+4\log\big(2N^2\big)+2\log^2\big(2N^2\big)}\bigg).
\end{align}

We observe that the first bound \eqref{eq:mainestimatealt1} is contained in the second \eqref{eq:mainestimatealt2} as the case $N=1$. Thus \eqref{eq:mainestimatealt2} holds for all $N\in \Z_{\geq 1}$. Then, the inequality
\begin{align}
\lvert p(z)\rvert \leq \lvert \re p(z)\rvert+\lvert \im p(z)\rvert= \lvert f(\theta)\rvert+\lvert (Hf)(\theta)| \quad (z\in \partial \D)
\end{align}
together with \eqref{eq:mainestimatealt2} implies the result \eqref{eq:corestimate}.

\subsubsection{Proof of Corollary~\ref{cor:complex}}
\label{subsec:complexproof}
Set
\begin{equation}
f(\theta)=\re p(\ee^{\ii\theta})=\frac 12 \sum_{n=0}^N \re a_n \big(\ee^{\ii n\theta}+\ee^{-\ii n\theta}\big)	-\frac1{2\ii}\sum_{n=1}^N \im a_n\big(\ee^{\ii n\theta}-\ee^{-\ii n\theta}\big);
\end{equation}
we have
\begin{equation}
\im f(\theta)=\frac 1{2\ii} \sum_{n=1}^N \re a_n\big(\ee^{\ii n\theta}-\ee^{-\ii n\theta}\big)	+\frac1{2}\sum_{n=0}^N \im a_n\big(\ee^{\ii n\theta}+\ee^{-\ii n\theta}\big). 
\end{equation}
Using \eqref{eq:Hexp}, we see
\begin{equation}
\im f(\theta)=-(H\re f)(\theta)	+\im a_0.
\end{equation}

The remainder of the proof is similar to that of Corollary~\ref{cor:main} in Section~\ref{sec:corproof}.

\end{document}